\documentclass[11pt]{article}

\pdfoutput=1

\usepackage{geometry}
\usepackage{graphicx}
\usepackage{amsfonts,amssymb,amsmath,amsthm,paralist}
\usepackage{color}
\usepackage{amscd}
\usepackage{nccmath}

\definecolor{libl}{cmyk}{0.2,0.1,0,0}

\newtheorem{Def}{Definition}
 
\newtheorem{Lemma}{Lemma}

\newtheorem{Theorem}{Theorem}

\newcommand{\ncd}{\newcommand}
\ncd{\lce}{\equiv_{l.C.}}
\ncd{\cg}{\mbox{cg}}
\ncd{\Tr}{\mbox{Tr}}

\newcommand{\nix}[1]{}

\begin{document}

\title{Measurement-based quantum computation--a quantum-mechanical toy model for spacetime?}

\author{R. Raussendorf, P. Sarvepalli, T.-C. Wei, P. Haghnegahdar\vspace{2mm}\\
{\small{\em{University of British Columbia, Department of Physics and Astronomy, Vancouver, BC, Canada}}}}

\maketitle

\begin{abstract}
  We propose measurement-based quantum computation (MBQC) as a quantum mechanical toy model for spacetime. Within this framework, we discuss the constraints on possible temporal orders enforced by certain symmetries present in every MBQC. We provide a classification for all MBQC temporal relations compatible with a given initial quantum state and measurement setting, in terms of a matroid. Further, we find a symmetry transformation related to local com\-ple\-men\-ta\-tion that leaves the temporal relations invariant. After light cones and closed time-like curves have previously been found to have MBQC counterparts, we identify event horizons as a third piece of the phenomenology of General Relativity that has an analogue in MBQC. 

\end{abstract}
 
\section{Introduction}

Unifying quantum mechanics \cite{BBR}-\cite{HMM} with general relativity \cite{AEGR, AEGR2} is a major open problem in physics. After the unification of quantum mechanics and special relativity has been accomplished \cite{DFT, DLI}, in the resulting relativistic quantum field theories \cite{QED1} - \cite{QCD2}, spacetime and quantum-mechanical degrees of freedom are still treated on a separate footing. For example, the electromagnetic, weak and strong forces are described as being mediated by particles, namely photons, $W$ and $Z$-bosons, and gluons. An analogous consistent formulation of gravity has so far not succeeded. In order to overcome the schism between the theory of gravity on one side and quantum mechanics on the other, one may therefore ask: Is there an element in the structure of quantum mechanics that forces spacetime into existence? I.e., is spacetime an emergent phenomenon, predicted by quantum mechanics?

In the present paper, we approach this question in a restricted and controlled setting, with emphasis on the notion of time. Namely, we study a scheme of quantum computation \cite{NC}, so-called `measurement-based' or `one-way' quantum computation (MBQC) \cite{RB01} as a quantum-mechanical toy model for spacetime. In MBQC, the process of computation is driven by local measurements on a suitably chosen entangled state. `Time' in deterministic MBQC is a partial ordering of measurement events, which is compatible with the notion of time in the theory of relativity. An external and continuously increasing clock time is not required. The possible temporal orders are severely constrained by symmetries present in every MBQC. The exploration of these symmetry-enforced constraints is a central theme in this paper. 

A fundamental property of quantum mechanics is that measurement outcomes---as opposed to measurement bases---cannot be chosen. They are in general random. MBQC, which is driven by measurements (no unitary evolution ever takes place in this computational scheme), were entirely useless if this randomness could not be prevented from creeping into the logical processing. But it can: all that is needed is the adjustment of measurement bases according to measurement outcomes obtained at other qubit locations {\em{before}}. One can turn this mechanism for adjusting measurement bases around, and say that it is the inherently probabilistic nature of quantum measurement which determines the `before' in MBQC. Indeed, without even knowing what is being computed in a given MBQC, the requirement of preventing the randomness of quantum mechanical measurement from affecting the logical processing imposes stringent constraints on the possible temporal orders of measurement events. It is because of this link between temporal order and quantum measurement that we propose MBQC as a quantum-mechanical toy model for spacetime.

In the context of the present work, it has previously been shown for the case the initial state in MBQC is a stabilizer state \cite{Stab} known in a suitable local basis, that the entire temporal order of measurements is fully specified if the set $I$ of first-measurable and the set $O$ of last-measurable qubits are known \cite{Gflow}, \cite{EffFlow}. But not every pair $I$, $O$ leads to a possible temporal order. A condition on admissible pairs $I$, $O$ has been given in \cite{Mhalla}. On a different route, an MBQC-analogue of closed time-like curves (CTCs) has been identified \cite{CTC}, via post-selection models \cite{BennettCTC}, \cite{SvetlichnyCTC} for CTCs in quantum circuits. Certain solutions of the Einstein equations permit closed time-like curves, and if MBQC is to serve as a toy model for a quantum spacetime it is a bonus that it accommodates this part of the phenomenology of general relativity.

The results of the present paper are three-fold. First, we identify a gauge transformation leaving the classical output of all MBQCs invariant, and causing additional terms in the classical processing relations for measurement outcomes. We show that the resulting generalized processing relations fully specify the resource state and measurement setting, up to equivalence. In the reverse direction, we provide a classification for all MBQC temporal relations compatible with a given resource stabilizer state and set of measurement bases. Specifically, we introduce a matroid capturing all information of the resource state and measurement bases relevant for temporal order, and show that the classical processing relations---including temporal order---are in one-to-one correspondence with the bases of this matroid. Second, we identify a further group of symmetry transformations generally changing MBQCs but preserving temporal order. The possible temporal orders in MBQC thus arise as representations of this group. These symmetry transformations turn out to be related to local complementation \cite{LC1}-\cite{LC3}. Third, after light cones and closed time-like curves, we identify event horizons as a third piece of GR phenomenology that has a counterpart in MBQC. 

\section{Background}
\label{BG}

In this section we review some basic facts about measurement-based quantum computation, essential definitions for the discussion of MBQC temporal order, as well as previous work in this area. The scheme of MBQC itself, with a proof of its computational universality, is not reviewed here since various articles on this subject exist in the literature \cite{RB01}, \cite{NielQCc} - \cite{Overview} Also, we require familiarity with the stabilizer formalism \cite{Stab}.

\subsection{MBQC and cluster states}

In MBQC, the process of computation is driven by local (=1-qubit) measurements on an initial highly entangled state, generally taken to be a so-called cluster state. The local measurements can only reduce entanglement, and therefore all entanglement needed for the computation must come from the initial state. For this reason, the initial state is often called the `resource state'\footnote{Large entanglement is a necessary \cite{vdN1, vdN2} but not sufficient condition for the usefulness of a quantum state as resource in MBQC. Somewhat paradoxically, quantum states exist that are too entangled to be useful \cite{Eisert1}.}. 

The computational power of a given MBQC strongly depends on the choice of the initial resource state. For example, a local resource state has obviously no computational power. Other states may be used for a restricted class of computations. Two-or-higher dimensional cluster states of unbounded size have the property that they enable universal quantum computation. That is, {\em{any}} quantum computation can be realized on such a state, by suitable choice of the local measurement bases.

We now define graph states, to be used later on, and cluster states as a subclass thereof.

\begin{Def}[Graph states and cluster states.] Be $G$ a graph with vertex set $V(G)$ and edge set $E(G)$, such that there is one qubit for each vertex $a \in V(G)$. Then, the graph state $|G\rangle$ is the unique (up to global phase) joint eigenstate, $|G\rangle = K_a\,|G\rangle$, for all $a \in V(G)$, of the operators
\begin{equation}
K_a = \sigma_x^{(a)}\bigotimes_{b |(a,b)\in E(G)}\sigma_z^{(b)}.
\end{equation}
The cluster state $|\Phi_{\cal{L}}\rangle$ is a graph state where the corresponding graph is a $d$-dimensional lattice ${\cal{L}}$.
\end{Def}
In the standard scheme \cite{RB01}, the measured observables are of the form
\begin{equation}
\label{ObsXY}
O_a[q_a]=\cos \varphi_a\, \sigma_x^{(a)} + (-1)^{q_a}\sin \varphi_a \,\sigma_y^{(a)},
\end{equation}
with $a$ the qubit to which the measurement is applied, and $q_a\in \mathbb{Z}_2$ depending on outcomes of (earlier) measurements on other qubit locations.

\subsection{Temporal order in MBQC}
\label{OTO}

As noted above, the temporal order of measurement in MBQC is a consequence of the randomness of measurement outcomes. By adjusting measurement bases according to measurement outcomes obtained on other qubits, this randomness inherent in  quantum mechanical measurement can be kept from creeping into the logical processing \cite{RB01}. If the measurement outcome of qubit $a$ influences the choice of measurement basis for qubit $b$, clearly, qubit $a$ must be measured before qubit $b$ can. This is how a temporal order among the measurement events arises.

We now generalize the above scenario of measuring observables of form Eq.~(\ref{ObsXY}) on cluster states. Namely, we now consider general stabilizer states $|\Psi\rangle$ as resources, with $\mbox{supp}(|\Psi\rangle)=\Omega$ and stabilizer group ${\cal{S}}(|\Psi\rangle)$. Furthermore, we generalize the local observables whose measurement drives the computation from Eq.~(\ref{ObsXY}) to
\begin{equation}
  \label{Obs}
  O_a[q_a] =  \cos \varphi_a\, \sigma_\phi^{(a)} +(-1)^{q_a} \sin \varphi_a\,  \sigma_{s\phi}^{(a)},\;\; \forall a \in \Omega,
\end{equation}
with $\sigma_\phi, \sigma_{s\phi} \neq \sigma_\phi \in \{X,Y,Z\}$. Therein, the measurement angles $\varphi_a$ are in the range $-\pi/2\leq \varphi_a <\pi/2$, and $q_a\in\mathbb{Z}_2$ may depend on measurement outcomes from several (other) qubits in $\Omega$. 
\begin{Def} [Measurement plane] 
For every qubit $a \in \mbox{supp}(|\Psi\rangle)$, the measurement plane at $a$ is the ordered pair $[\sigma_\phi^{(a)}, \sigma_{s\phi}^{(a)}]$. 
\end{Def}
We define a third Pauli operator, $\sigma_s = i \sigma_{s\phi}\sigma_\phi$. As we will see shortly, the Pauli operators $\sigma_\phi$ and $\sigma_s$ are useful because of the relations
\begin{equation}
  \label{obs2}
  \sigma_\phi O[q] \sigma_\phi^\dagger = O[q \oplus 1],\;\;  \sigma_s O[q] \sigma_s^\dagger = -O[q]. 
\end{equation} 
The basic mechanism of accounting for an ``undesired'' measurement outcome is the following. Suppose on some qubit $a \in \Omega$, instead of the ``desired'' post-measurement state $|\varphi_a\rangle_a$ the ``undesired'' post-measurement state $|\varphi_a^\perp\rangle_a$ has been obtained. The goal is to get the computation back on track by only adjusting the subsequent measurements. To do that, we require a stabilizer operator $\tilde{K}(a) \in {\cal{S}}(|\Psi\rangle)$ with the following properties \cite{Gflow}: (1) $\tilde{K}(a)$ has support only on $a$ and the yet unmeasured qubits, and (2) $\tilde{K}(a)|_a=\sigma_s^{(a)}$. Recall that $\sigma_s|\varphi_a\rangle = |\varphi_a^\perp\rangle$ for the eigenstates $|\varphi_a(q_a,s_a)\rangle$ of the local measured observable $O[q_a]$, c.f. Eq.~(\ref{obs2}).  Denote by ${\cal{P}}(a)$ and ${\cal{F}}(a)$ the past and future of $a$, respectively. Then,
\begin{equation}
\label{corr1}
\begin{array}{rcl}
  \left(\mbox{}_{{\cal{P}}(a)}\langle\varphi_{\text{loc}}|\otimes\,_a\langle \varphi_a^\perp|\right)|\Psi\rangle &=&  \left(\mbox{}_{{\cal{P}}(a)}\langle\varphi_{\text{loc}}|\otimes\,_a\langle \varphi_a^\perp|\right) \tilde{K}(a)|\Psi\rangle\\ 
& = & \left(\mbox{}_{{\cal{P}}(a)}\langle\varphi_{\text{loc}}|\otimes\,_a\langle \varphi_a|\right) \tilde{K}(a)\left|_{{\cal{F}}(a)}\right. |\Psi\rangle
\end{array}
\end{equation}
Therein, the first equality follows from $\tilde{K}(a) \in {\cal{S}}(|\Psi\rangle)$, and the second from the above properties (1) and (2).  

Since the overlaps between local states (representing the local measurements) with the resource state $|\Psi\rangle$ contain all information about the computation, we thus find that we can correct for ``undesired'' outcomes by (a) adjusting measurement bases of future measurements (caused by tensor factors $\sigma_\phi$ in $\tilde{K}(a)\left|_{{\cal{F}}(a)}\right.$) and (b) re-interpretation of measurement outcomes (caused by tensor factors $\sigma_s$ in $\tilde{K}(a)\left|_{{\cal{F}}(a)}\right.$).

{\em{Example:}} Consider MBQC on a cluster state $|\Phi_3\rangle$ of three qubits on a line, each measured in the $[\sigma_x,\sigma_y]$-plane. That is, for all three qubits $\sigma_s = Z$ and $\sigma_\phi = X$. (Here and from now on, we use the shorthand $X \equiv \sigma_x$, $Y \equiv \sigma_y$ and $Z \equiv \sigma_z$.) The stabilizer generators of $|\Phi_3\rangle$ are
\begin{equation}
  \label{3Cluster}
  \begin{array}{rclcl}
    K_1 &=& X_1\otimes Z_2 \otimes I_3 &=& \sigma_\phi^{(1)} \otimes \sigma_s^{(2)} \otimes I^{(3)},\\ 
    K_2 &=& Z_1 \otimes X_2 \otimes Z_3 &=& \sigma_s^{(1)} \otimes \sigma_\phi^{(2)} \otimes \sigma_s^{(3)},\\
    K_3 &=& I_1 \otimes Z_2 \otimes X_3 &=& I^{(1)} \otimes \sigma_s^{(2)} \otimes \sigma_\phi^{(3)}.
  \end{array}
\end{equation}
When the three cluster qubits are measured in the order $1\prec 2 \prec 3$, the corresponding quantum circuit is \cite{RBB03}
\begin{equation}
\label{EqCirc}
\parbox{5cm}{\includegraphics[width=5cm]{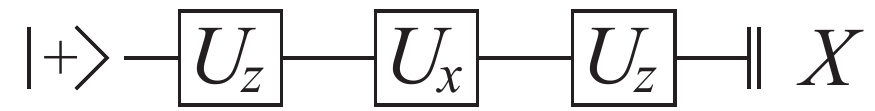}}
\end{equation}
Using Eq.~(\ref{corr1}), here we show that if qubits 1, 2, 3 are measured in the order $1\prec 2 \prec 3$, then the randomness of the measurement outcomes on qubits 1 and 2 can be corrected for. First we consider the stabilizer operator $K_2 = \sigma_s^{(1)} \otimes \sigma_\phi^{(2)} \otimes \sigma_s^{(3)}=:K(1)$. If inserted in the state overlap of Eq.~(\ref{corr1}), the measurement outcome of qubits 1 and 3 is flipped, as well as the measurement basis at qubit 2. Since qubits 2 and 3 are yet unmeasured when qubit 1 is measured, this is a valid correction operation for qubit 1; hence the notation $K(1)$. Similarly, $K_3 = I^{(1)} \otimes \sigma_s^{(2)} \otimes \sigma_\phi^{(3)} =: K(2)$ can be used as correction operation for qubit 2. $K(2)$ flips the measurement outcome of qubit 2 and the measurement basis of qubit 3. Since qubit 3 is yet unmeasured when qubit 2 is measured, this corresponds to a valid correction operation.

The above argument also works in reverse. If the correction operations $K(1)$ and $K(2)$ are used, then $1 \prec 2 \prec 3$ follows. $K(1)$ implies that the measurement basis of qubit 2 depends on the measurement outcome at qubit 1, hence $1 \prec 2$. Note that $1\prec 3$ does not yet follow! $K(1)$ does {\em{not}} affect the measurement basis at qubit 3. Only the meaning of the eigenstates is interchanged, which by itself  does not require qubit 3 to be measured after qubit 1. The interpretation of the measurement outcome may take place long after the measurement itself has taken place. 

$K_3=K(2)$ implies that the measurement basis of qubit 3 depends on the measurement outcome of qubit 2, and hence $2\prec 3$. Both relations taken together yield $1 \prec 2 \prec 3$.  

From the equivalence with the circuit of Eq.~(\ref{EqCirc}) one would expect one bit of classical output. Indeed, if no correction operations need to be used, the eigenvalue measured at the output of the circuit corresponds to the eigenvalue $\lambda_3$ measured on qubit 3 of the cluster. Now recall that $K(1)$, applied conditioned upon $\lambda_1=-1$, flips $\lambda_3$. Therefore, with or without corrections, the eigenvalue measured in the circuit Eq.~(\ref{EqCirc}) equals $\lambda_1\lambda_3$. Or, in binary notation $\lambda_1\equiv (-1)^{s_1},\, \lambda_3\equiv (-1)^{s_3}$, the single bit of classical output takes the value $s_1 + s_3 \mod 2$.

Note that we have not made any use of $K_1=\sigma_\phi^{(1)}\otimes \sigma_s^{(2)}$ in the above argument. Still, $K_1$ has a role to play, as we discuss in Section~\ref{gauge}.

\subsection{Influence matrix, forward and backward cones}

To counteract the randomness of measurement outcomes, two measurement settings (i.e., bases) per qubit suffice, which may be labeled by $q_a=0$ and $q_a=1$, respectively, for each qubit $a$. The measurement settings may collectively be described by a binary vector $\textbf{q}$, $[\textbf{q}]_a = q_a$ for all $a\in \Omega$, and the measurement outcomes by a binary vector $\textbf{s}$, $[\textbf{s}]_a = s_a$, forall $a \in \Omega$. It turns out that the relation between measurement bases $\textbf{q}$ and measurement outcomes \textbf{s} is {\em{linear}} \cite{RB02},
\begin{equation}
\label{TO1}
\textbf{q} = T \textbf{s} \mod 2,
\end{equation}
with $T$ a binary matrix. We call $T$ the influence matrix. 

The set of all qubits $b$ whose measurement basis must be adjusted according to the measurement outcome on $a$ is denoted as the {\em{forward cone}} of a qubit $a$. Similarly, the backward cone of a qubit $b$ is the set of all those qubits $a$ whose measurement outcome influence the measurement basis at $b$. More formally, with Eq.~(\ref{TO1}), 
\begin{Def}[Forward and backward cones]\label{FBC} For any $a \in \Omega$ the forward cone $\mbox{fc}(a)$ is given by
\begin{equation}
  \mbox{fc}(a):=\left\{ b \in \Omega|\, \partial  q_b / \partial s_a =1 \right\}.
\end{equation}
For any $b \in \Omega$ the backward cone $\mbox{bc}(b)$ is given by
\begin{equation}
  \mbox{bc}(b):=\left\{ a \in \Omega|\, \partial  q_b / \partial s_a =1 \right\}.
\end{equation}
\end{Def}

We denote the characteristic vectors of $fc(a)$ and $bc(b)$ by $\mbox{\textbf{fc}}(a)$ and $\mbox{\textbf{bc}}(b)$, respectively. Then, the influence matrix $T$ takes the form
\begin{equation}
\label{TO2}
  T = \left(\begin{array}{c} (\;\;\;\textbf{bc}(1)\;\;\;)\\
  (\;\;\;\textbf{bc}(2)\;\;\;)\\ .\\.\\(\;\;\;\textbf{bc}(n)\;\;\;) 
\end{array} \right) = 
   \left(\begin{array}{cccc} \!\!\!\left(\begin{array}{c}\mbox{ }\\ \mbox{ }\\ \!\!\!\!\!\textbf{fc}(1)\!\!\!\!\!\\ \mbox{ } \\ \mbox{ } \end{array} \right)
 \left(\begin{array}{c}\mbox{ }\\ \mbox{ }\\ \!\!\!\!\!\textbf{fc}(2)\!\!\!\!\!\\ \mbox{ } \\ \mbox{ } \end{array} \right) .. 
\left(\begin{array}{c}\mbox{ }\\ \mbox{ }\\ \!\!\!\!\!\textbf{fc}(n)\!\!\!\!\!\\ \mbox{ } \\ \mbox{ } \end{array} \right)
  \end{array} \!\!\!\right).
\end{equation}
The influence matrix $T$ generates a temporal relation among the measurement events under transitivity. We say $a \prec b$ ($a$ precedes $b$) if $b \in fc(a)$. A priori, it is not forbidden that for two qubits $a,b \in \Omega$, $a \prec b$ and $b \prec a$. However, such a computation could not be run deterministically in a world like ours where time progresses linearly. An MBQC is {\em{deterministically runnable}} if the temporal relation ``$\prec$'' between the measurement events is a strict partial order.  

\begin{Def}[Strict partial order]
A strict partial order is a relation among the elements $a \in \Omega$ with the following properties
\begin{equation}
  \begin{array}{llcr}
  a \not \prec a, & \forall a \in \Omega, &&(\mbox{\em{irreflexivity}})\\
  a \prec b \Longrightarrow b \not \prec a,& \forall a,b \in \Omega &&(\mbox{\em{antisymmetry}})\\
  a \prec b, b \prec c \Longrightarrow a \prec c, & \forall a,b,c \in \Omega. && (\mbox{\em{transitivity}})
  \end{array}
\end{equation}
\end{Def}

\begin{Def}[Input and output sets.]\label{IOdef} For a given MBQC, the input set $I \subseteq \Omega$ is the set of qubits whose backward cones are empty, $I = \{a \in \Omega|\,  bc(a)= \emptyset\}$. The output set $O \subseteq \Omega$ is the set of qubits whose forward cones are empty, $O = \{ b \in \Omega, fc(b) = \emptyset\}$.
\end{Def}
That is, with respect to a given temporal relation among the measurement events, $I$ is the maximal set of qubits which can be measured first, and $O$ is the maximal set of qubits which can be measured last.

Regarding the computational output, for the purpose of this paper we are exclusively interested in MBQCs for which the computational result is a classical bit string. That is, every qubit in $\Omega$ is measured. Then, as a consequence of the randomness of individual measurement outcomes, the classical output $\textbf{o}$ of an MBQC is given by {\em{correlations}} among measurement outcomes. Again, the relation between classical output and measurement outcomes is linear,
\begin{equation}
\label{TO1b}
\textbf{o} = Z \textbf{s} \mod 2,
\end{equation}
for a suitable binary matrix $Z$.\medskip

{\em{Example:}} To illustrate the above notions, we briefly return to the three-qubit cluster state example of Section~\ref{OTO}. From the previous discussion we find that
\begin{equation}
\label{TrelEx0}
\left(\begin{array}{c} q_1\\q_2\\q_3 \end{array}\right) = \left(\begin{array}{ccc} 0 & 0 & 0\\ 1 & 0 & 0 \\ 0 & 1 & 0\end{array}\right) \left(\begin{array}{c} s_1\\s_2\\s_3 \end{array}\right) \mod 2.
\end{equation}
Therefore, $fc(1)=\{2\}$, $fc(2)=\{3\}$, $fc(3)=\emptyset$ and $bc(1)=\emptyset$, $bc(2)=\{1\}$, $bc(3)=\{2\}$. Hence, $I=\{1\}$ and $O=\{3\}$. Also, $1\prec 2$ and $2\prec 3$. The latter two relations generate a third under transitivity, namely $1 \prec 3$. Asymmetry and irreflexivity are obeyed in this example. Regarding the single bit of output in this computation, the matrix $Z$ of Eq.~(\ref{TO1b}) is $Z=(1\,0\,1)$.

\subsection{Brief review of prior work on MBQC temporal order}

{\em{MBQC temporal order as an (almost) emergent phenomenon.}} In \cite{Gflow} the following question is asked: ``Given graph $G$ and the set $\Sigma$ of measurement planes for all vertices, can the temporal order of measurements in MBQC with a graph state $|G\rangle$ be uniquely reconstructed from this information?'' The graph $G$ and the measurement planes $\Sigma$ are undirected objects. Thus, if the answer to this question was yes, then temporal order in MBQC were truly emergent. 

However, it turns out that the pair $G,\Sigma$ does not specify the temporal order of measurements in MBQC uniquely; there are in general a number of consistent temporal orders respecting the requirement that the randomness of measurement outcomes should not affect the logical processing. One may then ask how constraining on temporal order this requirement actually is. To this question, the following answer is provided by \cite{Gflow}: If in addition to $G$ and $\Sigma$ the set $I$ of first-measurable and the set $O$ of last-measurable qubits is known, then the complete temporal order (if existing) can be uniquely reconstructed from this information. Thus, MBQC temporal order is not emergent in the strict sense; a seed $I,O$ must be provided in addition to $G$ and $\Sigma$, and the complete temporal order then follows.

But not every pair $I$, $O$ will lead to a consistent temporal order. The question that now arises is which pairs $I,O$ do. For the case where the stabilizer resource state is a graph state and all qubits are measured in the $[X,Y]$-plane\footnote{The former condition alone is not a restriction, since all stabilizer states are local Clifford equivalent to graph states \cite{Grassl}, but both conditions jointly are.} then the answer to this question is given in \cite{Mhalla}. Denote by $A_G$ the adjacency matrix of the graph $G$ describing the resource state $|G\rangle$, and by $A_G|_{O^c\times I^c}$ the submatrix of $A_G$ where the rows are restricted to $O^c:=\Omega\backslash O$ and the columns are restricted to $I^c$. Then, the pair $I,O$ leads to a partial order of measurement events in MBQC iff there exists a matrix $T$ such that $A_G|_{O^c\times I^c} T = I$, and $T$ is free of cycles (that is $T_{aa}=0,\forall a$, $T_{ab}T_{ba}=0,\forall a,b$, $T_{ab}T_{bc}T_{ca}=0,\forall a,b,c$, etc). The resulting temporal order is generated by $T$ under transitivity. \medskip

\begin{figure}[h]
\begin{center}
  \begin{tabular}{ll}
    a) & b)\\
    \parbox{7cm}{\includegraphics[width=7cm]{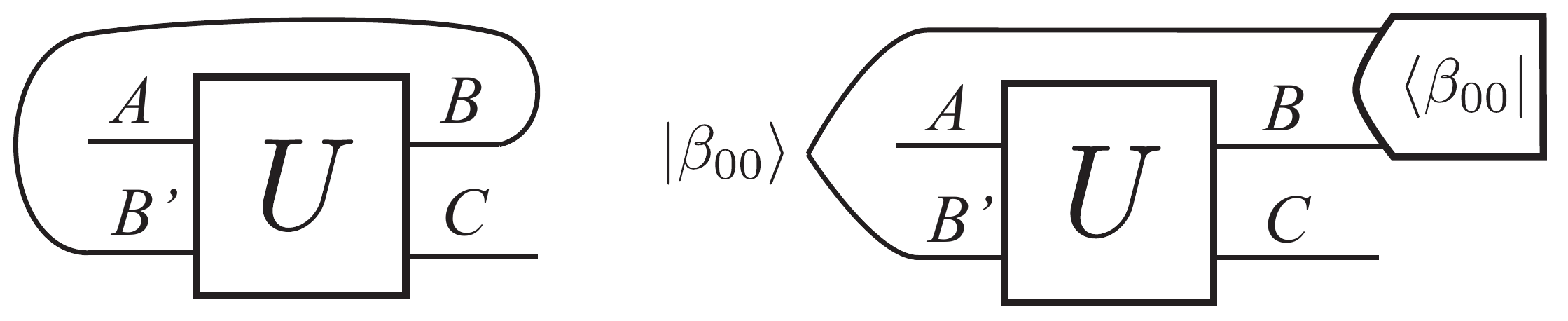}} & \parbox{7cm}{\includegraphics[width=7cm]{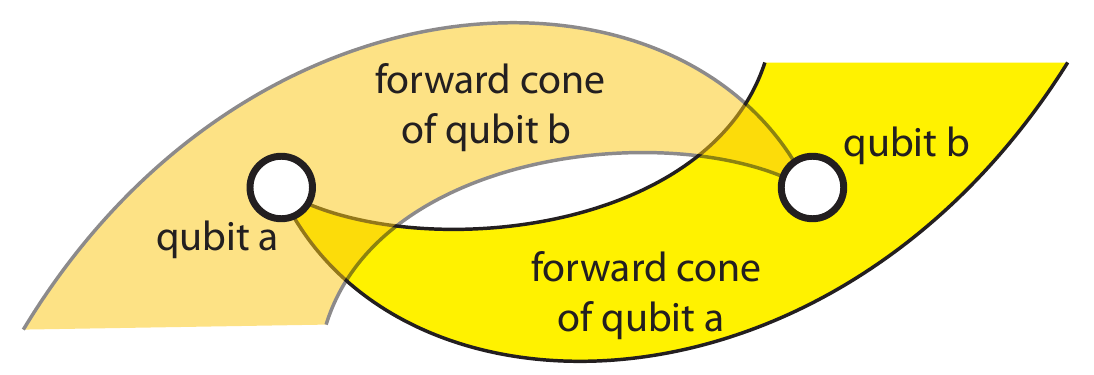}} 
  \end{tabular}
    \caption{\label{CTCfig} Closed time-like curves in MBQC. a) Bennett, Schumacher and Svetlichny's post-selection model \cite{BennettCTC,SvetlichnyCTC} of CTCs (left: circuit with wires `going backwards in time', right: implementation thereof using teleportation and post-selection). b) Nested forward cones in the MBQC equivalent of the teleportation circuit in (a).}
\end{center}
\end{figure}

\noindent
{\em{MBQC and closed time-like curves.}} In \cite{CTC} it is shown that MBQC encompasses the post-selection model of closed time-like curves (CTC's) proposed by Bennett, Schumacher \cite{BennettCTC} and Svetlichny \cite{SvetlichnyCTC}. The CTCs arise from circuits such as the one displayed in Fig.~\ref{CTCfig}a, translated into MBQC. The result are forward cones with the property $b \in fc(a)\,\wedge a \in fc(b)$, for two suitably chosen qubits $a,b \in \Omega$; see Fig.~\ref{CTCfig}b. Such nested forward cones are an obstruction to deterministic runnability of MBQC, but mimic closed time-like curves of General Relativity in the MBQC setting.

\section{Outline}

The remainder of the present paper is organized as follows. In Section~\ref{gauge} we show that every quantum computation---both in the circuit and the measurement-based model---transforms under a gauge transformation which affects the sequence of gates or measurements, respectively, but keeps the probability distribution of computational outputs invariant. For MBQC, this gauge transformation has the effect of causing extra terms in the classical processing relations Eq.~(\ref{TO1}), (\ref{TO1b}). 

In Sections~\ref{Intdep} and \ref{GTO} the consequences of these gauge transformations for MBQC are explored. It is already known \cite{Gflow} that the classical processing relations (containing the temporal order of measurement events) are not independent of the resource state and the set of measurement planes. For the extended processing relations---which now contain more information---this mutual dependence becomes stronger. In fact, the classical processing relations determine the resource state and set of measurement planes, c.f. Theorem~\ref{OneToOneTwo} in Section~\ref{Intdep}.

In Section~\ref{Flip} we introduce a second symmetry transformation acting on MBQCs which we call `flipping of the measurement plane'; see Fig~\ref{GT}b. As introduced in Section~\ref{BG}, the measurement plane $[\sigma_\phi,\sigma_{s\phi}]$ is an {\em{ordered}} pair of Pauli operators. Or is it? What happens under the exchange $\sigma_\phi \longleftrightarrow \sigma_{s\phi}$? (Under this transformation, the set of measurable observables in the plane is mapped back onto itself, hence the name for the transformation.) Denote by $a$ the qubit whose measurement plane is flipped. It turns out that whenever the measurement basis at $a$ does not depend on the measurement outcome at $a$, i.e., in the absence of a self-loop at $a$, flipping of the measurement plane at $a$ induces a transformation on the influence matrix $T$ which leaves the temporal order generated by $T$ invariant. The effect of flipping measurement planes on the influence matrix $T$ is closely related to the operation of local complementation \cite{LC1}-\cite{LC3} on graphs.

In the present context, we consider MBQC as a toy model for spacetime and therefore admit temporal relations with closed time-like curves. However, closed time-like curves are an obstacle to deterministic MBQC. In Section~\ref{BRKctc} we show that every MBQC solving a problem in the complexity class NP can, without reducing its probability of success, be transformed into a deterministically runnable MBQC, i.e., into an MBQC whose temporal relation of measurement events is a strict partial order. This transformation does not guard against the inefficiencies of post-selection, but restores deterministic runnability.

In Section~\ref{Concl} we conclude and point out open questions.

\section{Gauge degrees of freedom}
\label{gauge}

\begin{figure}[t]
  \begin{center}
    \begin{tabular}{ll}
      a) & b)\\
      \includegraphics[width=5.5cm]{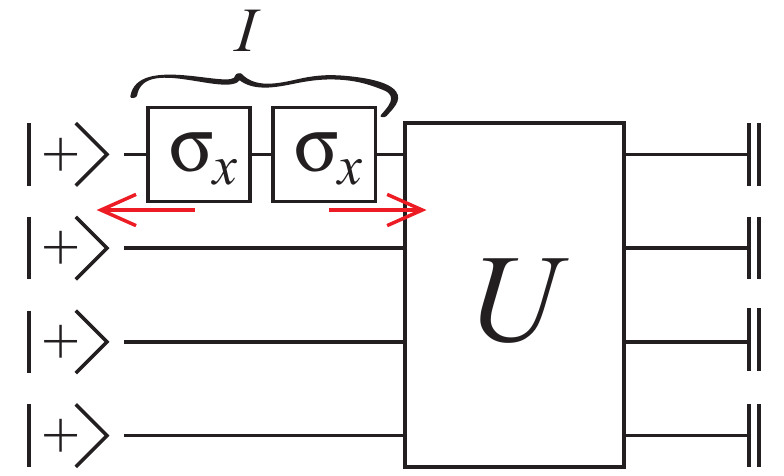} & 
      \includegraphics[width=8cm]{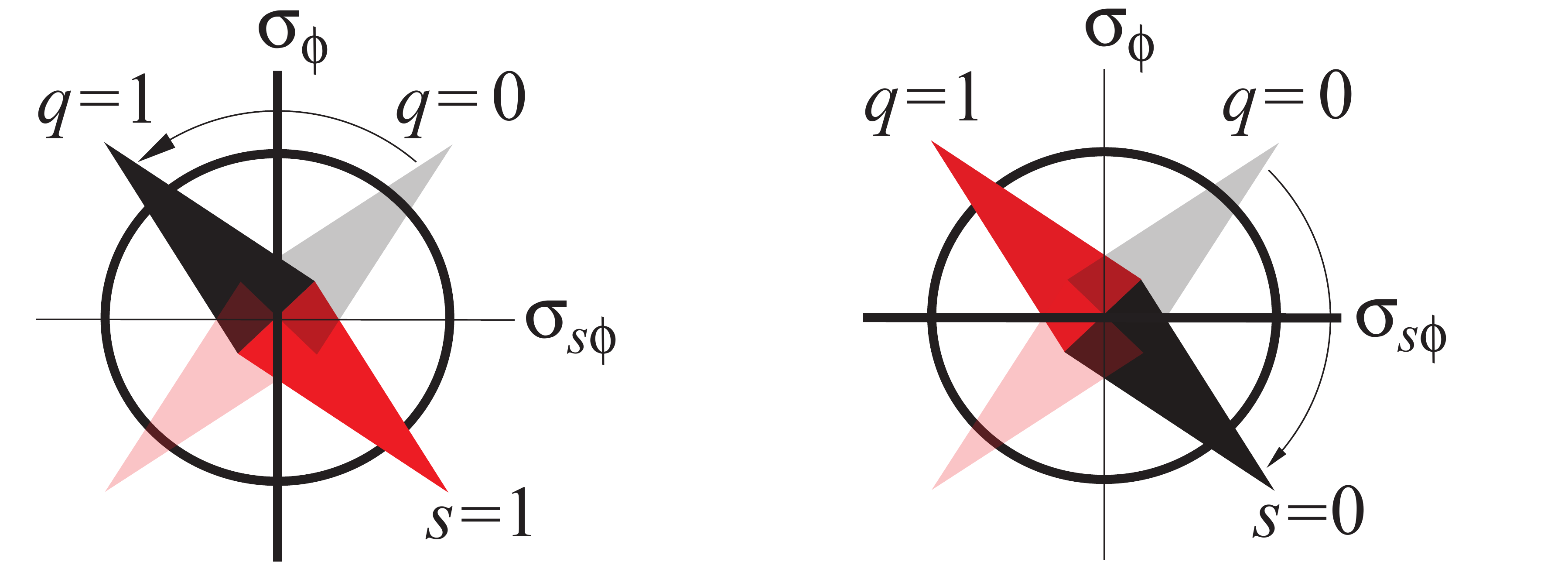}
    \end{tabular}
    \caption{\label{GT} Two symmetry transformations. a) Gauge transformation in the  circuit model. For any logical qubit $l$, an identity $I=\sigma_x^{(l)}\sigma_x^{(l)}$ is inserted into the circuit next to the input. The left $\sigma_x$ is propagated backwards in time, and absorbed by the input state $|+\rangle$. The right $\sigma_x$ is propagated forward in time, flipping rotation angles and, potentially, measurement outcomes in the passing. b) Flipping a measurement plane in MBQC. In the measurement plane $[\sigma_\phi,\sigma_{s\phi}]$, the Pauli operator $\sigma_\phi$ is distinguished over $\sigma_{s\phi}$ because the rule for adjusting a local observable $O_a$ for measurement is $O_a[q_a=1] = \sigma_\phi^{(a)} O_a[q_a=0] \big(\sigma_\phi^{(a)}\big)^\dagger$. If $T_{aa}=0$ (with $T$ the influence matrix) then the flip $\sigma_\phi^{(a)} \longleftrightarrow \sigma_{s\phi}^{(a)}$ is a symmetry transformation for the given MBQC.}
  \end{center}
\end{figure}

Here we introduce the notion of ``gauge transformations'' acting on a given quantum computation. These transformations exist for both the circuit model and MBQC.

\subsection{Gauge transformations in the circuit model}
\label{GC}

To obtain an intuition for the gauge transformations introduced here, it is instructive to first inspect them in the circuit model. Specifically, we consider a quantum circuit which consists of (1) the preparation of the quantum register in the initial state $\bigotimes_{i = 1}^n|+\rangle_i$, (2) unitary evolution composed of, say, CNOT gates and one-qubit rotations about the $X$- and $Z$-axes, and (3) local measurements for readout. Such a circuit is displayed in Fig.~\ref{GT} above. Then, into every qubit line individually,  we may insert an identity $I = \sigma_x \sigma_x$ next to the input; See Fig.~\ref{GT}. The left $\sigma_x$ is propagated backwards in time until absorbed by the input state $|+\rangle$. The right $\sigma_x$ is propagated forward in time, flipping rotation angles and readout-measurement outcomes in the passing. 

This transformation is an {\em{equivalence transformation}}, since it is caused by the insertion of an identity gate into the circuit of Fig.~\ref{GT}. The reason we call it a gauge transformation is that it changes reference frames for certain rotation angles. Specifically, for the rotation gates next to each input qubit we can individually choose our convention for which rotation angles are called positive or negative, respectively. Once those signs are fixed on the input side, they are fixed throughout the circuit.

\subsection{The gauge transformations in the measurement-based model}

Translating the above discussed gauge transformations from the circuit model into MBQC it is easily seen that the above relations Eq.~(\ref{TO1}) and (\ref{TO1b}) are incomplete. We find the more general relations
\begin{subequations}
  \begin{align}
  \label{TO7a}
  \textbf{q} &= T \textbf{s} + H \textbf{g} \mod 2,\\
  \label{TO7b}
  \textbf{o} &= Z \textbf{s} + R\textbf{g} \mod 2,
  \end{align}
\end{subequations}
with $\textbf{g}$ a choice of gauge. 

Under a change of $\textbf{g}$, the measurement bases in a particular MBQC change, but the probability distribution for the classical output values remains unchanged. Of course, knowledge of the $\textbf{g}$-de\-pen\-dent extra parts in Eq.~(\ref{TO7a},\ref{TO7b}) is not necessary to run any given MBQC, since $\textbf{g}=\textbf{0}$ is always a valid choice. However, the presence of the extra terms in the processing relations strengthens their interdependence with the resource state, which is the reason why we discuss them here.

We may now want to derive the generalized processing relations Eq.~(\ref{TO7a},\ref{TO7b}) directly in MBQC, without reference to the circuit model. Before discussing the general case, we return to the specific example of the three-qubit cluster state in Section~\ref{OTO}. 

{\em{Example:}} Consider the product $K_1K_3=\sigma_\phi^{(1)} \otimes \sigma_\phi^{(3)}$ of $K_1$, $K_3$ in Eq.~(\ref{3Cluster}). When used in Eq.~(\ref{corr1}), the effect of this stabilizer element is to flip the measurement bases of qubits 1 and 3. Hence, the relation Eq.~(\ref{TrelEx0}) generalizes to
\begin{equation}
\label{TrelEx}
\left(\begin{array}{c} q_1\\q_2\\q_3 \end{array}\right) = \left(\begin{array}{ccc} 0 & 0 & 0\\ 1 & 0 & 0 \\ 0 & 1 & 0 \end{array}\right) \left(\begin{array}{c} s_1\\s_2 \\s_3 \end{array}\right) + \left(\begin{array}{c} 1\\ 0 \\ 1 \end{array}\right)g_1 \mod 2.
\end{equation}
This is precisely what one would expect from the insertion of $\sigma_x\sigma_x$ next to the input of the equivalent quantum circuit in Eq.~(\ref{EqCirc}). 

We now turn to the general case. Via Eq.~(\ref{corr1}), the stabilizer group ${\cal{S}}(|\Psi\rangle)$ acts on $\textbf{s}$ and $\textbf{q}$. Denote the post-measurement state of qubit $a \in \Omega$ by $|s_a,q_a\rangle_a$, with $s_a$ the measurement outcome and $q_a$ specifying the measurement basis. Then, as in Eq.~(\ref{corr1}),
\begin{equation}\label{GaugeId}
  \left(\bigotimes_{a\in \Omega}\mbox{}_a\langle s_a,q_a|\right)|\Psi\rangle = \left(\bigotimes_{a\in \Omega}\mbox{}_a\langle s_a,q_a|\right)K|\Psi\rangle = \left(\bigotimes_{a\in \Omega}\mbox{}_a\langle s_a,q_a|K\right)|\Psi\rangle,\;\;\;\;\;\;\; \forall K \in {\cal{S}}(|\Psi\rangle).
\end{equation}
Since $\sigma_s|s,q\rangle = |s\oplus 1, q\rangle$ and $\sigma_\phi|s,q\rangle = |s, q\oplus 1\rangle$, for a stabilizer element $K=\bigotimes_{a \in \Omega}(\sigma_s^{(a)})^{v_a}(\sigma_\phi^{(a)})^{w_a}$ the action of $G_K$ on $\textbf{s}$, $\textbf{q}$ is 
\begin{equation}
  \label{G1}
  \begin{array}{rl}
  G_{K}: & \begin{array}{rcl} \textbf{s} &\longrightarrow& \textbf{s} + \textbf{v} \mod 2,\\
 \textbf{q} &\longrightarrow& \textbf{q} + \textbf{w} \mod 2.\end{array}
  \end{array}
\end{equation} 
Again, nothing changes by the insertion of a stabilizer (identity) operator into the state overlap of Eq.~(\ref{corr1}), and transformations $G_K$ are therefore equivalence transformations. They can be used to constrain the possible temporal orders in MBQC, as we discuss explicitly in Section~\ref{GTO}.\medskip

\subsection{Correction and gauge operations in the stabilizer formalism}

To state and prove our results on MBQC temporal order, we need to make a few more definitions. First, let us emphasize that in this paper we allow closed time like curves (CTC) in the computation; we are interested in classifying all transitive \emph{temporal relations} consistent with a given resource stabilizer state and set of measurement planes. Therefore, we do not restrict to partial orders per se, and drop the conditions of irreflexivity and antisymmetry on the temporal relations.

It turns out that the possible temporal relations, and indeed the classical processing relations Eq.~(\ref{TO7a}), (\ref{TO7b}), can be parametrized by two subsets of $\Omega$, namely the computational output set $O_{\text{comp}}$ and the gauge input set $I_{\text{gauge}}$. We define these sets next. 

We say that the measurement outcome $s_a$ of qubit $a$ is {\em{corrected}} in a given MBQC, if by insertion of a suitable stabilizer operator $K$ in Eq.~(\ref{GaugeId}) equivalence of $s_a=1$ with the reference outcome $s_a=0$ is established, at the cost of adjustment of measurement bases and/or re-interpretation of measurement outcomes on other qubits.

We observe that if for a given MBQC {\em{all}} measurement outcomes $s_a$, $a \in \Omega$, can be corrected then these measurement outcomes contain no information and no linear combination of them is worth outputting. Thus, in general there will be a set of qubits whose measurement outcomes are not corrected.

\begin{Def}[Computational output set]\label{COS} For a given MBQC, the computational output set $O_{\text{comp}}\subseteq \Omega$ is the set of qubits whose measurement outcomes are not corrected.
\end{Def}
The correction operations for the qubits in $a \in (O_{\text{comp}})^c$, c.f. Eq.~(\ref{corr1}), are each implemented by correction operators $K(a) \in {\cal{S}}(|\Psi\rangle)$. For each $a \in  (O_{\text{comp}})^c$, $K(a)$ has the property that it flips the measurement outcome $s_a$ at qubit $a$, but does not flip the measurement outcome on any other qubit in $(O_{\text{comp}})^c$. In this way, it is ensured that the correction operation is for qubit $a$ individually. 

\begin{Def}[Correction operator]\label{CorrOp} For an MBQC with a given stabilizer resource state $|\Psi\rangle$, fixed set $\Sigma$ of measurement planes and computational output set $O_{\text{comp}}$, for each $a \in (O_{\text{comp}})^c$ the corresponding correction operator $K(a) \in {\cal{S}}(|\Psi\rangle)$ is a Pauli operator satisfying the conditions
\begin{equation}
  \label{CorrCorr}
  \begin{array}{rclll}
    K(a)|_a & \in & \{\sigma_s^{(a)}, \sigma_{s\phi}^{(a)}\},\\
    K(a)|_b & \in & \{I^{(b)},\sigma_\phi^{(b)}\},& \forall b \in (O_{\text{comp}})^c\backslash a.
  \end{array}
\end{equation}
\end{Def}
The correction operator $K(a)$ can be used to correct an ``undesired'' measurement outcome at qubit $a$, c.f. Eq.~(\ref{corr1}). If, for a given operator $K(a)$, $K(a)|_b \in  \{\sigma_\phi^{(b)}, \sigma_{s\phi}^{(b)}\}$ then, by Eq.~(\ref{corr1}), the measurement basis at qubit $b$ depends on the measurement outcome for qubit $a$. In terms of the influence matrix $T$,
\begin{equation}
  \label{KTrel}
K(a)|_b \in \{\sigma_\phi^{(b)},\sigma_{s\phi}^{(b)}\}  \Longleftrightarrow [T]_{ba} = 1,\;\;\;\;\forall a \in (O_{\text{comp}})^c,\, b\in \Omega.  
\end{equation}
Note that in the first line of Eq.~(\ref{CorrCorr}) we allow $K(a)|_a =\sigma_{s\phi}^{(a)}$ only because we are admitting closed time-like curves in the present discussion. $K(a)|_a =\sigma_{s\phi}^{(a)}$ means that the measurement basis for qubit $a$ depends on the outcome $s_a$ of the measurement of qubit $a$. This amounts to a closed time-like curve only involving qubit $a$ (self-loop) and is an obstacle to deterministic runnability.

Because the qubits $a \in O_{\text{comp}}$ have no correction operations, $T_{ba}=0$ for all $b \in \Omega$, and thus 
\begin{equation}
\label{OOc}
O_{\text{comp}} \subseteq O.
\end{equation}  
Thus, $O_{\text{comp}}$ and the correction operators $\{K(a), a \in (O_{\text{comp}})^c\}$ completely determine the influence matrix $T$ and hence the temporal relation among the measurements. 

Conversely, if $O_{\text{comp}}$ and $\{K(a), a \in (O_{\text{comp}})^c\}$ are unknown, then the constraints Eq.~(\ref{CorrCorr}) pose self-consistency conditions on them. We discuss these conditions further below.\medskip

As we have seen in the concrete three-qubit example above, for a given MBQC the relation between the choice $\textbf{q}$ of measurement bases and the measurement outcomes $\textbf{s}$ in general allows for an offset term $H\textbf{g}$, c.f. Eq.~(\ref{TO7a}). Thus, the measurement bases for a certain set of qubits can be freely chosen until the initially arbitrary $\textbf{g}$ becomes fixed. This observation leads to 

\begin{Def}[Gauge input set]\label{GIS}For a given MBQC with input set $I$, the gauge input set $I_{\text{gauge}}\subseteq I$ is a set of qubits such that for each $i \in I_{\text{gauge}}$ the parameter $q_i$ specifying the locally measured observable $O_i[q_i]$ can be freely chosen. 
\end{Def}

Analogously to the correction operations, there will be gauge operators which implement the `corrections' of measurement {\em{bases}} for the qubits in $I_{\textbf{gauge}}$. The definition of the gauge operators ensures that for all $i \in I_{\text{gauge}}$ the corresponding $q_i$ can be changed individually without changing the others.

\begin{Def}[Gauge operators]\label{GaugeOp} For an MBQC with a given stabilizer resource state $|\Psi\rangle$, fixed set $\Sigma$ of measurement planes, computational output set $O_{\text{comp}}$ and gauge input set $I_{\text{gauge}}$, for each $i \in I_{\text{gauge}}$, the corresponding gauge operator $\overline{K}(i) \in {\cal{S}}(|\Psi\rangle)$ is a Pauli operator satisfying the conditions
\begin{equation}
\label{CorrGauge}
\begin{array}{rcll}
  \overline{K}(i)|_i &=& \sigma_\phi^{(i)},\\
  \overline{K}(i)|_j &=& I^{(j)},& \forall j \in (I_{\text{gauge}} \cap (O_{\text{comp}})^c) \backslash i,\\
  \overline{K}(i)|_k &\in & \{I^{(k)},\sigma_\phi^{(k)}\}, & \forall k \in ((I_{\text{gauge}})^c \cap (O_{\text{comp}})^c) \backslash i, \\
\overline{K}(i)|_l &\in & \{I^{(l)},\sigma_s^{(l)}\}, & \forall l \in (I_{\text{gauge}} \cap O_{\text{comp}}) \backslash i. 
\end{array}
\end{equation}
\end{Def}
Like previously for the correction operations, Eq.~(\ref{CorrGauge}) poses self-consistency condition on the possible sets $I_{\text{gauge}}$ and $\{\overline{K}(i), i \in I_{\text{gauge}}\}$.\medskip 

\begin{Lemma}
\label{indepL}
The correction operators $K(a)$ of Eq.~(\ref{CorrCorr}), $a \in (O_{\text{comp}})^c$, and the gauge operators $\overline{K}(i)$ of Eq.~(\ref{CorrGauge}), $i \in I_{\text{gauge}}$, are independent. That is, for all sets $J \subset I_{\text{gauge}}$, $L \subset (O_{\text{comp}})^c$ with $J \neq \emptyset\,\vee\, L \neq \emptyset$, $\tilde{K}(J,L):=\prod_{i \in J}\overline{K}(i)\,\prod_{a \in L}K(a) \neq I^{(\Omega)}$.
\end{Lemma} 

{\em{Proof of Lemma~\ref{indepL}.}} (indirect) Assume that there exists a pair $J,L$, with $J \neq \emptyset\,\vee\, L \neq \emptyset$, such that $\tilde{K}(J,L) = I^{(\Omega)}$. Then, for all $b \in (O_{\text{comp}})^c$, $\tilde{K}(J,L)|_b=1$. Now, with Eq.~(\ref{CorrCorr}), $K(b)|_b \in \{\sigma_s^{(b)}, \sigma_{s\phi}^{(b)}\}$, and $K(c)|_b \in  \{\sigma_\phi^{(b)}, I^{(b)}\}$ for all $c \in (O_{\text{comp}})^c\backslash b$. With Eq.~(\ref{CorrGauge}), $\overline{K}(i)|_b \in \{\sigma_\phi^{(b)}, I^{(b)}\}$ for all $i \in I_{\text{gauge}}$. Therefore, no other $K(\cdot)$, $\overline{K}(\cdot)$ can cancel a $\sigma_s^{(b)}$-contribution from $K(b)$ to $\tilde{K}(J,L)$. Hence, $b \not\in L$ for all $b \in (O_{\text{comp}})^c$, and thus $L = \emptyset$. By an analogous argument, no $K(\cdot)$ can cancel the $\sigma_\phi^{(i)}$-contribution of $\overline{K}(i)$ to $\tilde{K}(J,L)$, hence $i \not \in J$ for all $i \in I_{\text{gauge}}$, and $J=\emptyset$. Thus, $\tilde{K}(J,L) = I^{(\Omega)} \Longrightarrow J,L=\emptyset$. Contradiction. Hence, the $K(a)$, $\overline{K}(i)$ are independent. $\Box$ \medskip

{\em{Remark:}} For any given MBQC with fixed temporal relation, the computational output set $O_{\text{comp}}$ is a subset of the output set $O$, see Eq.~(\ref{OOc}). But $O_{\text{comp}}$ and $O$ are not necessarily equal. Likewise, $I_{\text{gauge}}\subseteq I$ by definition, but $I_{\text{gauge}}$ and $I$ may not be equal. To illustrate this point, we consider the following two examples.

{\em{Example 1.}} Consider the 3-qubit cluster state of Section~\ref{OTO}, with measurement planes $[X/Y]$ for all three qubits. We consider the correction operators $K(1)=K_2$, $K(2)=K_3$ for qubits 1 and 2, and gauge operator $\overline{K}(1)=K_1K_3$.  $O_{\text{comp}}=\{3\}$ and $I_{\text{gauge}}=\{1\}$ are then permitted by Eqs.~(\ref{CorrCorr}) and (\ref{CorrGauge}), respectively. As was discussed previously for the above choice of correction operations, $I=\{1\}$ and $O=\{3\}$. Thus, in the present example $I_{\text{gauge}}=I$ and $O_{\text{comp}}=O$. 

{\em{Example 2.}} Consider a Greenberger-Horne-Zeilinger state $|GHZ\rangle = (|000\rangle + |111\rangle)/\sqrt{2}$ as resource state, with all three qubits measured in a basis in the $[X,Y]$-plane. We use the stabi\-lizer elements $K(1):=Z_1Z_3=\sigma_s^{(1)}\sigma_s^{(3)}$ and $K(2):=Z_2Z_3=\sigma_s^{(2)}\sigma_s^{(3)}$ as correction operations for qubits 1 and 2, and $\overline{K}(1):=X_1X_2X_3 = \sigma_\phi^{(1)}\sigma_\phi^{(2)}\sigma_\phi^{(3)}$ as gauge operator. Then, the choice $I_{\text{gauge}}=\{1\}$ and $O_{\text{comp}}=\{3\}$ is admitted by Eqs.~(\ref{CorrCorr}) and (\ref{CorrGauge}), respectively. On the other hand, this is an example of a temporarily flat MQC, $T=0$. Therefore, $I=O=\{1,2,3\}$, and $O_{\text{comp}}\neq O$, $I_{\text{gauge}} \neq I$.

\begin{Lemma}\label{sizes}
For any MBQC on a stabilizer state, $|I_{\text{gauge}}| \leq |O_{\text{comp}}|$. 
\end{Lemma}
We prove Lemma~\ref{sizes} in Section~\ref{stabGNF}.

\begin{Def}
A pair $I_{\text{gauge}}$, $O_{\text{comp}}$ is called extremal iff $|I_{\text{gauge}}| =|O_{\text{comp}}|$.
\end{Def}
As will become clear in the next section, extremal pairs $I_{\text{gauge}}$, $O_{\text{comp}}$ are easier to handle than general pairs, and are not very restrictive (c.f. Theorem~\ref{Extr}).\medskip

We still need to relate the classical output vector $\textbf{o}$ appearing in Eq.~(\ref{TO7b}) to the set $O_{\text{comp}}$. To this end, we make the following

\begin{Def}[Optimal classical output]\label{oco} A classical output vector $\textbf{o}$, with processing relations $\textbf{o} = Z \textbf{s} + R\textbf{g}$, is optimal iff the following conditions hold  
\begin{enumerate}
\item{{\em{Maximality:}} Upon left-multiplication by an invertible matrix, $Z$ can be brought into a unique normal form 
\begin{equation}
  \label{ZN}
  Z \sim \left({\tt{Z}}|I\right),
\end{equation}  
where the column split is between $(O_{\text{comp}})^c|O_{\text{comp}}$, and}
\item{{\em{Determinism:}} For the matrix ${\tt{Z}}$ in Eq.~(\ref{ZN}),
\begin{equation}
\label{ZD}
[{\tt{Z}}]_{ij}=1 \Longleftrightarrow K(j)|_i \in \{ \sigma_s^{(i)}, \sigma_{s\phi}^{(i)}\}.
\end{equation}}
\end{enumerate}
\end{Def}
In other words Eq.~\eqref{ZD} informs us that the $j$th column of ${\tt{Z}}$ is simply the 
restriction of the  support of $K(j)$ to $O_{\text{comp}}$.

The reason for defining an `optimal classical output' besides a `classical output' is the following: One could, in principle, run an MBQC perfectly deterministically and then choose $\textbf{o}$ such that nothing is outputted at all, or all outputted bits, independent of the measurement angles chosen, are zero guaranteed or perfectly random guaranteed. The above definition of an optimal classical output eliminates such choices. Maximality says that there is one bit of optimal output per qubit of $O_{\text{comp}}$. The determinism condition can be understood from the correction procedure explained in Eq.~(\ref{corr1}). For a qubit $j \in (O_{\text{comp}})^c$, we account for the undesired outcome $s_j=1$ by inserting $K(j)$ into the state overlap $\langle \varphi_{\text{loc}}|\Psi\rangle = \langle \varphi_{\text{loc}}|K(j)|\Psi\rangle$. Now consider a qubit $i \in O_{\text{comp}}$. By Eq.~(\ref{ZN}), $s_i$ contributes to the output bit $o_i$, the $i$-th bit of $\textbf{o}$. If $K(j)|_i \in \{\sigma_s^{(i)},\sigma_{s\phi}^{(i)}\}$, then the insertion of $K(j)$ into the overlap flips $s_i$, i.e., $s_i \longrightarrow s_i \oplus 1$. This needs to be taken into account when reading out $s_i$. The linear combination $s_j \oplus s_i$ remains unaffected by the correction for $s_j$. 

We have now made the definitions needed to state and prove our results on temporal order in MBQC. In Section~\ref{stabGNF} we establish a normal form for the stabilizer generator matrix of the resource state $|\Psi\rangle$. This normal form is the basis for our results on MBQC temporal order, which are stated in Sections~\ref{Intdep} and \ref{GTO}.

\subsection{Normal form of the resource state stabilizer}
\label{stabGNF}  

Let us briefly review which pieces of information specify an MBQC. A priori, there are four: the set of measurement angles, the set $\Sigma$ of measurement planes, the resource state $|\Psi\rangle$ and the classical processing relations Eq.~(\ref{TO7a},\ref{TO7b}). The measurement angles entirely drop out of all our considerations about temporal order. Next, we observe that when specifying the measurement planes and the resource state separately, we really specify too much. Starting from a given pair $|\Psi\rangle ,\Sigma$ of resource state and set of measurement planes, for any local Clifford unitary $U$, the pair $U|\Psi\rangle, U(\Sigma)$ obtained by applying $U$ to both the measurement planes $\Sigma$ and the stabilizer state $|\Psi\rangle$ is again a valid pair, i.e., it consists of a set of a stabilizer state and a set of measurement planes. Furthermore, it amounts to exactly the same computation as the original pair. The pair $|\Psi\rangle ,\Sigma$ is thus redundant. To remove this redundancy, we combine the measurement planes and the stabilizer state $|\Psi\rangle$ into the stabilizer generator matrix ${\cal{G}}(|\Psi\rangle)$ in the $\sigma_\phi/\sigma_s$-stabilizer basis, 
\begin{equation}
\label{Stab}
{\cal{G}}(|\Psi\rangle) = (\Phi||S). 
\end{equation}
Therein, the columns to the left (right) form the $\sigma_\phi$ ($\sigma_s$-) part of the stabilizer generator matrix. ${\cal{G}}(|\Psi\rangle)$ comprises all information from the resource state $|\Psi\rangle$ and the set of measurement planes $\Sigma$ relevant for the discussion of MBQC.  We do not need to know the state and the measurement planes separately. 

Now note that the correction operators $K(a)$, $a \in (O_{\text{comp}})^c$, and the gauge operators $\overline{K}(i)$, $i \in I_{\text{gauge}},$ are all elements of the stabilizer ${\cal{S}}(|\Psi\rangle)$, and, by Lemma~\ref{indepL}, are independent. Thus, they either form or can be completed to a set of generators for ${\cal{S}}(|\Psi\rangle)$. This observation leads us to the following
\begin{Lemma}\label{NF}
For any MBQC on a stabilizer state $|\Psi\rangle$ with extremal $I_{\text{gauge}}$, $O_\text{comp}$, the generator matrix ${\cal{G}}$ of ${\cal{S}}(|\Psi\rangle)$ can be written in the normal form
 \begin{equation}
\label{GNF}
\begin{array}{c} \mbox{ } \\ \mbox{ } \\ \\ {\cal{G}} \cong\end{array} 
 \begin{array}{c}
    \sigma_\phi\hspace{1.8cm}\sigma_s\mbox{ }\\
    \mbox{}\;I_{\text{gauge}}\;\;(I_{\text{gauge}})^c\; (O_{\text{comp}})^c\; O_{\text{comp}}\\ \\
  \left(\begin{array}{c|c||c|c}
    \parbox{0.9cm}{\begin{center}$0$\end{center}} &  \parbox{0.9cm}{\begin{center}${\tt{T}}^T$\end{center}} & \parbox{0.9cm}{\begin{center}$I$\end{center}} & \parbox{0.9cm}{\begin{center}${\tt{Z}}^T$\end{center}}\\ \hline
    \parbox{0.9cm}{\begin{center}$I$\end{center}} & \parbox{0.9cm}{\begin{center}${\tt{H}}^T$\end{center}} & \parbox{0.9cm}{\begin{center}$0$\end{center}} & \parbox{0.9cm}{\begin{center}${\tt{R}}^T$\end{center}}
  \end{array}\right).
\end{array}
\end{equation}
The matrices $\tt{H}$, $\tt{R}$, $\tt{T}$, $\tt{Z}$ are related to the matrices $H$, $R$, $T$, $Z$ governing the classical processing of measurement outcomes in MBQC, $\textbf{q} = T \textbf{s} + H \textbf{g} \mod 2$ and $\textbf{o} = Z \textbf{s} + R \textbf{g} \mod 2$, via
\begin{equation}
  \label{HRTZshape}
   T =  \left(\begin{array}{c|c} 0 & 0 \\ \hline  {\tt{T}} & 0 \end{array}\right), \;
\begin{array}{rccc}
H  =  \left(\begin{array}{c} I \\ \hline {\tt{H}} \end{array}\right),\; Z = \big(\, {\tt{Z}}\, |\, I\,\big)
\end{array},\, R = \tt{R}.
\end{equation} 
\end{Lemma}

\begin{proof}
By Definition~\ref{CorrOp}, a correction operator $K(a)$ exists for every $a\in (O_{\text{comp}})^c$. By~Eq.~(\ref{CorrCorr}) these operators have no support in  $I \supseteq I_{\text{gauge}}$ and are $|(O_{\text{comp}})^c|$ in number and must take the following form:
$\left(\begin{array}{c|c||c|c} 0 & A& I & B\end{array} \right)$, where the $\sigma_{\phi}$ part is 
split as $I_{\text{gauge}}$, $I_{\text{gauge}}^c$ while the $\sigma_s$-part is split as 
$O_{\text{comp}}$, $(O_{\text{comp}})^c$.
The gauge  operators are, by definition, of the form: $\left(\begin{array}{c|c||c|c} I & C& 0 &D \end{array} \right)$, where the column splits are as for the correction operators. 
Since there are $|\Omega|$ generators for the stabilizer ${\cal{G}}(|\Psi\rangle) $, of which
$|(O_{\text{comp}})^c|$ are already accounted for, there can be at most 
$|O_{\text{comp}}|$ independent  gauge operators, i.e., $|I_{\text{gauge}}|\leq |O_{\text{comp}}|$ which proves Lemma~\ref{sizes}.
In the present setting, the pair $I_{\text{gauge}}$, $O_{\text{comp}}$ is extremal by assumption, 
thus these two sets of generators exhaust the stabilizer generators and we can write the stabilizer
as 
\begin{equation}
\label{GN2}
\begin{array}{c} \mbox{ } \\ \mbox{ } \\ \\ {\cal{G}}(|\Psi\rangle) = \end{array} 
 \begin{array}{c}
    \mbox{}\;I_{\text{gauge}}\;\;(I_{\text{gauge}})^c\; (O_{\text{comp}})^c\; O_{\text{comp}}\\ \\
  \left(\begin{array}{c|c||c|c}
    \parbox{0.9cm}{\begin{center}$0$\end{center}} &  \parbox{0.9cm}{\begin{center}$A$\end{center}} & \parbox{0.9cm}{\begin{center}$I$\end{center}} & \parbox{0.9cm}{\begin{center}$B$\end{center}} \\ \hline
     \parbox{0.9cm}{\begin{center}$I$\end{center}} & \parbox{0.9cm}{\begin{center}$C$\end{center}} & \parbox{0.9cm}{\begin{center}$0$\end{center}} & \parbox{0.9cm}{\begin{center}$D$\end{center}}
  \end{array}\right),
\end{array}
\end{equation}
for suitable matrices $A$, $B$, $C$ and $D$. 
We now need to identify these matrices. By definition, $I_{\text{gauge}} \subseteq I$.  Measurement outcomes on qubits $a \in O_{\text{comp}}$ are not corrected, hence $fc(a)=\emptyset$ for all $a \in O_{\text{comp}}$, and $O_{\text{comp}} \subseteq O$ follows from the definition of the output set $O$. Then, the influence matrix $T$ takes  the form
\begin{equation}
\label{Tshape}
T =  \left(\begin{array}{c|c} 0 & 0 \\ \hline  {\tt{T}} & 0 \end{array}\right),
\end{equation}
where the column split is $(O_{\text{comp}})^c|O_{\text{comp}}$ and the row split is $I_{\text{gauge}}|(I_{\text{gauge}})^c$. 
Now consider the correction operator  $K(a)$ for $a \in (O_{\text{comp}})^c$ in the upper part of ${\cal{G}}(|\Psi\rangle)$ in Eq.~(\ref{GN2}). We already know from Eq.~\eqref{KTrel}, that the $\sigma_{\phi}$ part of $K(a)$ is the forward cone of $a$. Therefore we must have  $A={\tt{T}}^T$. Further comparing, with Eq.~\eqref{ZD}
which states that the restriction of the $\sigma_{s}$ part of the correction operator $K(a)$ to 
$O_{\text{comp}}$ is the  $a$th column of $\tt{Z}$. But this is precisely the $a$th row of $B$, thus 
we infer that  $B={\tt{Z}}^T$.

Next, consider row $i$ of the lower part of ${\cal{G}}(|\Psi\rangle)$ in Eq.~(\ref{GN2}). Row $i$ is $(0,..,0,1,0,..,0|\textbf{c}^T||\textbf{0}|\textbf{d}^T)$. The corresponding stabilizer operator $\overline{K}(i)$, when inserted into the overlap $\langle \varphi_{\text{loc}}|\Psi\rangle$ as in Eq.~(\ref{corr1}), flips the measurement basis at qubit $i$ and of qubits $l \in (I_{\text{gauge}})^c$ with $[\textbf{c}]_l=1$. It further flips the measurement outcomes at qubits $m \in O_{\text{comp}}$ with $[\textbf{d}]_m=1$. Therefore,
$$
\begin{array}{lcl}
  H =  \left(\begin{array}{c} I \\ \hline {\tt{H}} \end{array}\right),
  & \mbox{with} & {\tt{H}} = C^T,\; \mbox{and}\\
  R = {\tt{R}}, & \mbox{with} & {\tt{R}} = D^T.
\end{array}
$$
We thus arrive at the normal form Eq.~(\ref{GNF}). 
\end{proof}

\section{Interdependence of resource state and temporal order}
\label{Intdep}

Let us return to our discussion from the beginning of Section~\ref{stabGNF}, on which pieces of information are needed to describe an MBQC. At this stage, apart from the set of measurement angles which do not enter our discussion, we remain with two pieces of data specifying a given MBQC, namely ${\cal{G}}(|\Psi\rangle)$ and the processing relations Eq.~(\ref{TO7a}), (\ref{TO7b}). But it doesn't stop there. As the results of \cite{Gflow}, \cite{EffFlow}, \cite{Mhalla} show, the resource state $|\Psi\rangle$, the measurement planes $\Sigma$ and the influence matrix $T$---being part of the classical processing relations Eq.~(\ref{TO7a}), (\ref{TO7b})---are not independent. Specifically, given the sets $I$ of first-measurable and $O$ of last-measurable qubits in addition to $|\Psi\rangle$ and $\Sigma$, the temporal order (generated by $T$) can be worked out completely.

Here we prove a statement about the interdependence of ${\cal{G}}(|\Psi\rangle)$ and the MBQC classical processing relations which goes the opposite direction. Namely we show that the classical processing relations Eq.~(\ref{TO7a}), (\ref{TO7b}) uniquely specify the stabilizer generator matrix ${\cal{G}}(|\Psi\rangle)$, i.e., the pair $|\Psi\rangle$, $\Sigma$ up to equivalence; See Theorem~\ref{OneToOneTwo} below. Thus, only two pieces of data are needed to specify an MBQC that satisfies the determinism constraints, namely the measurement angles and the classical processing relations for the measurement outcomes.

A further question is whether the temporal relations compatible with a resource state $|\Psi\rangle$ and set of measurement planes $\Sigma$ fit into a common framework. In this regard, we show that the classical processing relations (containing the temporal order) for MBQC with a fixed resource state and set of measurement planes, for extremal pairs $I_{\text{gauge}}, O_{\text{comp}}$, are in one-to-one correspondence with the bases of  ${\cal{G}}(|\Psi\rangle)$, c.f. Theorem~\ref{OneToOneThree}. 

\subsection{Results}

We now present four theorems on the mutual dependence of the resource state and the classical processing relations. 

\begin{Theorem}
  \label{OneToOneOne}
  Consider an MBQC on a stabilizer state $|\Psi\rangle$, with fixed measurement planes and an extremal pair of gauge input set $I_{\text{gauge}}$ and computational output set $O_{\text{comp}}$. Then, the relations $\textbf{q} = T \textbf{s} + H \textbf{g} \mod 2$, and $\textbf{o} = Z \textbf{s} + R \textbf{g} \mod 2$ for an optimal output $\textbf{o}$ are unique.
\end{Theorem}
That is, once the resource state $|\Psi\rangle$, the measurement planes and $I_{\text{gauge}}$, $O_{\text{comp}}$ are fixed, there is no freedom left to choose the classical processing relations. They are uniquely determined by the former. In particular, for fixed stabilizer state $|\Psi\rangle$ and measurement planes, $T=T(I_{\text{gauge}},O_{\text{comp}})$, $H=H(I_{\text{gauge}},O_{\text{comp}})$ etc. 

A corollary of Theorem~\ref{OneToOneOne} is that given the measurement planes and an extremal pair $I_{\text{gauge}}$, $O_{\text{comp}}$, the resource state $|\Psi\rangle$ uniquely determines the influence matrix $T$. One may ask how restrictive a condition the extremality of the pair $I_{\text{gauge}}$, $O_{\text{comp}}$ is. In this regard, note

\begin{Theorem}\label{Extr} Consider an MBQC on a fixed resource stabilizer state for fixed measurement planes, with an influence matrix $T$ and input and output sets $I(T)$, $O(T)$, such that no qubit $a \in I^c$ can be individually gauged with respect to $I(T)$, $O(T)$. Then, there exists an extremal pair $I_{\text{gauge}} \subseteq I$, $O_{\text{comp}} \subseteq O$ such that $T = T(I_{\text{gauge}},O_{\text{comp}})$.
\end{Theorem}
The input set $I(T)$ and the output set $O(T)$ which appear in Theorem~\ref{Extr} are uniquely specified by $T$ through Definition~\ref{IOdef}. $T(I_{\text{gauge}},O_{\text{comp}})$ is uniquely specified by the pair $I_{\text{gauge}}$, $O_{\text{comp}}$ through Theorem~\ref{OneToOneOne}.

Theorem~\ref{Extr} states that all temporal relations for an MBQC, subject to the extra condition on the qubits which can be individually gauged, arise from extremal pairs $I_{\text{gauge}}$, $O_{\text{comp}}$. 

By establishing Theorem~\ref{Extr} we trade the condition of the pairs $I_{\text{gauge}}$, $O_{\text{comp}}$ being extremal for the condition that no qubit in $I^c$ can be individually gauged. The latter is a more meaningful condition. Suppose a qubit $a$ in $I^c$ could be individually gauged wrt $I$, $O$. Then $\overline{K}(a)$ exists. For any $b$ with $K(b)|_a=\sigma_\phi^{(a)}$, $\tilde{K}(b) :=K(b)\overline{K}(a)$ is a valid correction operator for qubit $b$, and $\tilde{K}(b)|_a = I^{(a)}$. Hence, $a$ could be removed from all forward cones and thereby be made a qubit in $I$. By imposing the extra condition in Theorem~\ref{Extr}, we exclude temporal relations where certain qubits could be in the input set $I$ but aren't.\medskip 

Theorem~\ref{OneToOneOne} is mute on the question of which extremal pairs $I_{\text{gauge}}$, $O_{\text{comp}}$ are admissible. Theorem~\ref{OneToOneThree} below describes how much freedom remains for the choice of the classical processing relations, given ${\cal{G}}(|\Psi\rangle)$.

\begin{Theorem}\label{OneToOneThree} For MBQC with a fixed resource stabilizer state $|\Psi\rangle$ and fixed measurement planes, the classical processing relations for extremal $I_{\text{gauge}}$, $O_{\text{comp}}$, as specified by the matrices $H$, $R$, $T$, $Z$ and the sets $I_{\text{gauge}}, O_{\text{comp}}$, are in one-to-one correspondence with the bases of the matroid ${\cal{G}}(|\Psi\rangle)$.  
\end{Theorem}

After we have justified our restriction to extremal pairs of gauge input and computational output sets in Theorem~\ref{Extr} and have characterized the set of temporal relations compatible with a given resource state and set of measurement planes in Theorem~\ref{OneToOneThree}, we now return to Theorem~\ref{OneToOneOne}, and show that a converse also holds. 
\begin{Theorem}
\label{OneToOneTwo}
Consider an MBQC on a stabilizer state $|\Psi\rangle$, with classical processing relations $\textbf{q} = T \textbf{s} + H \textbf{g} \mod 2$, $\textbf{o} = Z\textbf{s} + R\textbf{g} \mod 2$ for an optimal classical output $\textbf{o}$, such that $\text{rk}\, H = \text{rk}\,Z$. Then the classical processing relations uniquely specify the stabilizer generator matrix ${\cal{G}}(|\Psi\rangle)$ in the $\sigma_\phi/\sigma_s$-basis, i.e. the resource stabilizer state $|\Psi\rangle$ and set $\Sigma$ of measurement planes up to equivalence.
\end{Theorem}

\subsection{The MBQC - spacetime correspondence}

Our proposal of MBQC as a toy model for spacetime is inspired by Malament's theorem \cite{conTC}, stating that for a continuous spacetime manifold, the metric is determined by the temporal order of spacetime events up to a conformal factor. Our Theorem~\ref{OneToOneTwo} shows close resemblance with Malament's theorem if we assert the following correspondence:
\begin{equation}
\label{corresp}
\parbox{10cm}{\includegraphics[width=10cm]{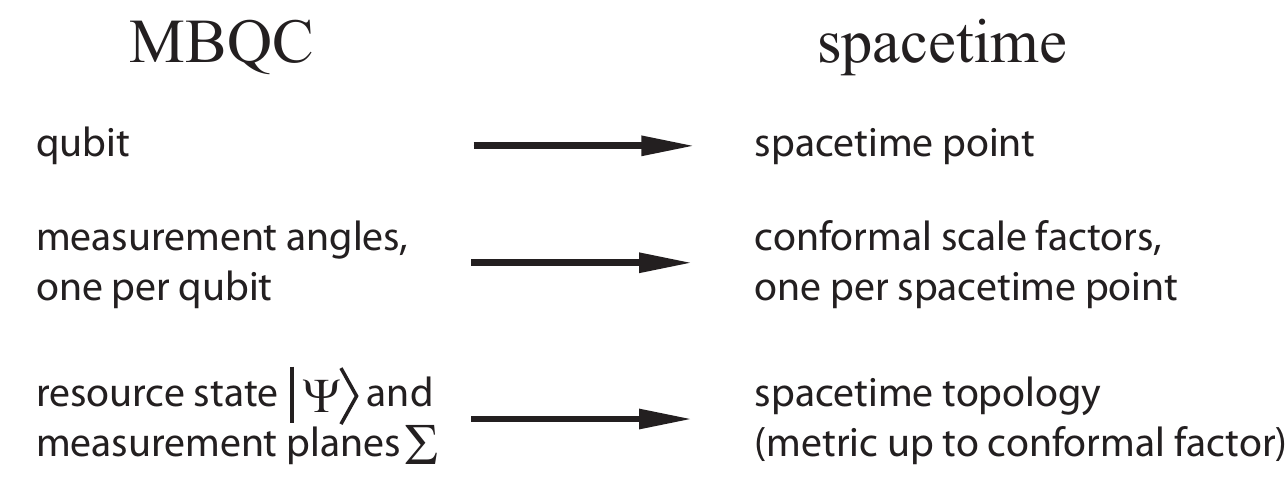}}
\end{equation}  
While the first line of the correspondence (\ref{corresp}) appears intuitive, the second line really comes about by identifying the left-over pieces on either side. Nevertheless, it seems a match: A reel-valued parameter, the measurement angle at a given qubit is mapped to a real-valued parameter, the scale factor at the corrsponding spacetime point. (Note, though, that the domain of the measurement angle is a compact set while that of the scale factor is not). A physical interpretation for this correspondence remains to be found. 

The support for the conjectured correspondence (\ref{corresp}) comes from its third line. By analogy with Malament's theorem, the temporal order of measurement events should, up to equivalence at least, determine the resource state $|\Psi\rangle$ and the set $\Sigma$ of measurement planes. Theorem~\ref{OneToOneTwo} almost gives this result. But the correspondence seems imperfect: It is not the temporal order alone---contained in the influence matrix $T$---which yields the pair $|\Psi\rangle$, $\Sigma$ up to equivalence. Rather it is the entire classical processing relations Eq.~(\ref{TO7a}), (\ref{TO7b}), specified by the four matrices $T$, $H$, $R$, $Z$.

But we can do better! Namely, it can be shown that the four matrices $T$, $H$, $R$, $Z$ can be assembled to an influence matrix $T_{\text{ext}}$ of an MBQC on a slightly bigger resource state $|\Psi'\rangle$. Specifically, the support of $|\Psi'\rangle$ is, in comparison to $|\Psi\rangle$, enlarged at the temporal boundaries $I_{\text{gauge}}$, $O_{\text{comp}}$. By this extension, all information contained in the classical processing relations becomes temporal information.

To obtain $T_{\text{ext}}$ we first choose a pair $I_{\text{gauge}} \subset I, O_{\text{comp}} \subset O$ ($I,O$ follow from $T$), and extract from $H,R,T,Z$ the corresponding matrices ${\tt{H}}, {\tt{R}}, {\tt{T}}, {\tt{Z}}$ wrt to the normal form Eq.~(\ref{GNF}) specified by $I_{\text{gauge}},O_{\text{comp}}$.  We assemble the composite matrix  
\begin{equation}
\label{TppL}
T_{\text{ext}} =  \left(\begin{array}{c|c|c|c} 0 & 0 & 0 & 0\\ I & 0 & 0 &0 \\ \hline {\tt{H}} & {\tt{T}}  & 0 & 0\\ \hline {\tt{R}} & {\tt{Z}} & I & 0 \end{array}\right),
\end{equation}
which is a square matrix of size $(|\Omega|+|I_{\text{gauge}}|+|O_{\text{comp}}|) \times (|\Omega|+|I_{\text{gauge}}|+|O_{\text{comp}}|)$. 

$T_{\text{ext}}$ is the influence matrix for MBQC on a bigger resource state $|\Psi'\rangle$ constructed from $|\Psi\rangle$. For the support $\Omega'$ of $|\Psi'\rangle$ we require two additional sets of qubits, $I'$ and $O'$, with $|I'|=|O'|=|I_{\text{gauge}}|=|O_{\text{comp}}|$, such that $\Omega' = \Omega \cup I' \cup O'$. $|\Psi'\rangle$ is obtained from $|\Psi\rangle$ by the following construction:
\begin{equation}
\label{PsiExt}
\parbox{5.5cm}{\includegraphics[width=5.5cm]{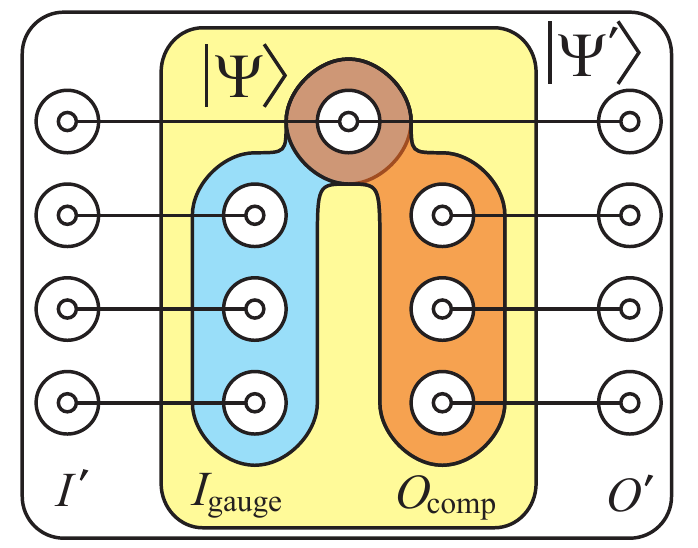}}
\end{equation} 
Therein, the gates $\parbox{1cm}{\includegraphics[width=1cm]{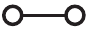}}=\Lambda_{s}$ are $\sigma_s$-controlled $\sigma_s$-gates, i.e., $\Lambda_s \sigma_s^{(1)} \Lambda_s^\dagger = \sigma_s^{(1)}$, $\Lambda_s \sigma_\phi^{(1)} \Lambda_s^\dagger = \sigma_\phi^{(1)} \otimes \sigma_s^{(2)}$, etc, and the extra qubits in $I'$ and $O'$ are initially prepared in the eigenstate of $\sigma_\phi$ with eigenvalue 1. With the definition Eq.~(\ref{PsiExt}) of $|\Psi'\rangle$, the labelling of the blocks of rows and columns for the matrix on the r.h.s. of Eq.~(\ref{TlcExt}) is $I'|(O_{\text{comp}})^c | O_{\text{comp}} | O'$ for the columns and $I' | I_{\text{gauge}} | (I_{\text{gauge}})^c| O'$ for the rows. 

We thus find that all information in the classical MBQC processing relations is temporal information for the computation on a slightly extended resource state.\medskip

\subsection{Proofs of Theorems~\ref{OneToOneOne}-\ref{OneToOneTwo}} 

Theorem~\ref{OneToOneOne} is an immediate consequence of Lemma~\ref{NF}. \medskip

\begin{proof}[Proof of Theorem~\ref{Extr}]
Assume that $I$ is valid input set and $O$ is a valid output set for a given MBQC. Then, the stabilizer generator matrix of the resource state can be written in the $\sigma_\phi/\sigma_s$-basis as
\begin{equation}
\label{GN3}
\begin{array}{c} \mbox{ } \\ \mbox{ } \\ {\cal{G}}(|\Psi\rangle) = \end{array} 
 \begin{array}{c}
    \mbox{}\;I \hspace{1cm} I^c\hspace{1cm} O^c\hspace{0.85cm} O\\ 
  \left(\begin{array}{c|c||c|c}
    \parbox{0.9cm}{\begin{center}$0$\end{center}} &  \parbox{0.9cm}{\begin{center}${\tilde{\tt{T}}}^T$\end{center}} & \parbox{0.9cm}{\begin{center}$I$\end{center}} & \parbox{0.9cm}{\begin{center}$A$\end{center}} \\ \hline
     \parbox{0.9cm}{\begin{center}$B$\end{center}} & \parbox{0.9cm}{\begin{center}$C$\end{center}} & \parbox{0.9cm}{\begin{center}$0$\end{center}} & \parbox{0.9cm}{\begin{center}$D$\end{center}}
  \end{array}\right),
\end{array}
\end{equation}
for some matrices $A$, $B$, $C$ and $D$. The influence matrix $T$ can be obtained as
\begin{equation}
  \label{T}
  T^T = \left( \begin{array}{c|c} 0 & {\tilde{\tt{T}}}^T \\ \hline 0 & 0\end{array} \right),
\end{equation}
with the column split between $I$ and $I^c$, and the row split between $O^c$ and $O$. 
The matrices $B$ and $(B|C)$ do not necessarily have maximal row-rank. By row transformations of 
$(B|C||0|D)$ we extract the dependent rows, and obtain
\begin{equation}
\label{GN3b}
\begin{array}{c} \mbox{ } \\ \mbox{ } \\  {\cal{G}}(|\Psi\rangle) = \end{array} 
 \begin{array}{c}
    \mbox{}\;I \hspace{1cm} I^c\hspace{1cm} O^c\hspace{0.85cm} O\\ 
  \left(\begin{array}{c|c||c|c}
    \parbox{0.9cm}{\begin{center}$0$\end{center}} &  \parbox{0.9cm}{\begin{center}${\tilde{\tt{T}}}^T$\end{center}} & \parbox{0.9cm}{\begin{center}$I$\end{center}} & \parbox{0.9cm}{\begin{center}$A$\end{center}} \\ \hline
     \parbox{0.9cm}{\begin{center}$0$\end{center}} & \parbox{0.9cm}{\begin{center}$0$\end{center}} & \parbox{0.9cm}{\begin{center}$0$\end{center}} & \parbox{0.9cm}{\begin{center}$D'$\end{center}}
\\ \hline
     \parbox{0.9cm}{\begin{center}$0$\end{center}} & \parbox{0.9cm}{\begin{center}$C''$\end{center}} & \parbox{0.9cm}{\begin{center}$0$\end{center}} & \parbox{0.9cm}{\begin{center}$D''$\end{center}}
\\ \hline
     \parbox{0.9cm}{\begin{center}$B'''$\end{center}} & \parbox{0.9cm}{\begin{center}$C'''$\end{center}} & \parbox{0.9cm}{\begin{center}$0$\end{center}} & \parbox{0.9cm}{\begin{center}$D'''$\end{center}}
  \end{array}\right),
\end{array}
\end{equation}
However, note that each row in the third set of rows in the above matrix can either be interpreted as a correction operator $K(a)$, with non-empty forward cone $fc(a)$,  for some $a \in O$, or a gauge operator $\overline{K}(i)$ for some $i\in I^c$. The former is ruled out because every $a \in O$ must have an empty forward cone.  The latter  is ruled out by the assumption that no qubit outside the set $I$ can be individually gauged. Therefore that set of rows must identically vanish.

Since ${\cal{G}}(|\Psi\rangle)$ has full row rank, so does the matrix $D'$ appearing in Eq.~(\ref{GN3b}). We may then choose a set $\Delta O \subseteq O$ such that the columns of $D'$ indexed by $\Delta O$ form a maximal independent set. We set $O_{\text{comp}}:=O\backslash \Delta O$.  
Then, by further row transformations which do not affect $\tilde{\tt{T}}$, the matrix in Eq.~(\ref{GN3b}) can be converted to
\begin{equation}
\label{GN4}
\begin{array}{c} \mbox{ } \\ \mbox{ } \\  {\cal{G}}(|\Psi\rangle) = \end{array} 
 \begin{array}{c}
    \mbox{ }\mbox{ }\;I \hspace{1cm} I^c\hspace{1cm} O^c\hspace{0.7cm} \Delta O \hspace{0.5cm} O_{\text{comp}}\\ 
  \left(\begin{array}{c|c||c|c|c}
    \parbox{0.9cm}{\begin{center}$0$\end{center}} &  \parbox{0.9cm}{\begin{center}${\tilde{\tt{T}}}^T$\end{center}} & \parbox{0.9cm}{\begin{center}$I$\end{center}} & \parbox{0.9cm}{\begin{center}$0$\end{center}} & \parbox{0.9cm}{\begin{center}$A'$\end{center}} \\ \hline
     \parbox{0.9cm}{\begin{center}$0$\end{center}} & \parbox{0.9cm}{\begin{center}$0$\end{center}} & \parbox{0.9cm}{\begin{center}$0$\end{center}} &  \parbox{0.9cm}{\begin{center}$I$\end{center}} & \parbox{0.9cm}{\begin{center}$A''$\end{center}} \\ \hline
     \parbox{0.9cm}{\begin{center}$B'''$\end{center}} & \parbox{0.9cm}{\begin{center}$C'''$\end{center}} & \parbox{0.9cm}{\begin{center}$0$\end{center}} &  \parbox{0.9cm}{\begin{center}$0$\end{center}} & \parbox{0.9cm}{\begin{center}$D'''$\end{center}}
  \end{array}\right).
\end{array}
\end{equation}
$B'''$ has full row rank by construction. We can therefore find a set $I_{\text{gauge}} \subseteq I$ such that the columns of $B'''$ indexed by $I_{\text{gauge}}$ form a maximal independent set. For any such set $I_{\text{gauge}}$ we can convert the matrix in Eq.~(\ref{GN4}) fully into the normal form Eq.~(\ref{GNF}) without affecting $T$.

For any of the above choices for $I_{\text{gauge}} \subseteq I$ and $O_{\text{comp}}\subseteq O$, the resulting influence matrix $T(I_{\text{gauge}}, O_{\text{comp}})$ can be extracted as 
\begin{equation}
  \label{Tpr}
  T(I_{\text{gauge}}, O_{\text{comp}})^T = \left( \begin{array}{c|c} 0 & \tilde{\tt{T}}^T \\ \hline 0 & 0  \\ \hline 0 & 0\end{array} \right),
\end{equation}
with the column split between $I$ and $I^c$, and the row split between $O^c$, $\Delta O$ and $O_{\text{comp}}=O\backslash \Delta O$. By comparison of Eqs.~(\ref{T}) and (\ref{Tpr}) we verify $T = T(I_{\text{gauge}},O_{\text{comp}})$. 
\end{proof}

{\em{Remark:}}
Comparing Eq.\eqref{GN4} with the normal form in Eq.~\eqref{GNF}, we can write the stabilizer matrix in a slightly varied form that can be useful later.
\begin{eqnarray}
\label{GN5}
\begin{array}{c} \mbox{ } \\ \mbox{ } \\ {\cal{G}}(|\Psi\rangle) = \end{array} 
 \begin{array}{c}
    I_{\text{gauge}}\hspace{0.7cm} \Delta I \hspace{0.7cm} I^c\hspace{0.8cm} O^c\hspace{0.7cm} \Delta O \hspace{0.5cm} O_{\text{comp}}\\ 
  \left(\begin{array}{c|c|c||c|c|c}
 \parbox{0.9cm}{\begin{center}$0$\end{center}} &    \parbox{0.9cm}{\begin{center}$0$\end{center}} &  \parbox{0.9cm}{\begin{center}${\tilde{\tt{T}}}^T$\end{center}} & \parbox{0.9cm}{\begin{center}$I$\end{center}} & \parbox{0.9cm}{\begin{center}$0$\end{center}} & \parbox{0.9cm}{\begin{center}${\tt{Z}}_1^T$\end{center}} \\ \hline
 \parbox{0.9cm}{\begin{center}$0$\end{center}} &     \parbox{0.9cm}{\begin{center}$0$\end{center}} & \parbox{0.9cm}{\begin{center}$0$\end{center}} & \parbox{0.9cm}{\begin{center}$0$\end{center}} &  \parbox{0.9cm}{\begin{center}$I$\end{center}} & \parbox{0.9cm}{\begin{center}${\tt{Z}}_2^T$\end{center}} \\ \hline
 \parbox{0.9cm}{\begin{center}$I$\end{center}} &     \parbox{0.9cm}{\begin{center}${\tt{H}}_1^T$\end{center}} & \parbox{0.9cm}{\begin{center}${\tt{H}}_2^T$\end{center}} & \parbox{0.9cm}{\begin{center}$0$\end{center}} &  \parbox{0.9cm}{\begin{center}$0$\end{center}} & \parbox{0.9cm}{\begin{center}${\tt{R}}^T$\end{center}}
  \end{array}\right),
\end{array}
\end{eqnarray}
where ${\tt{Z}} = ({\tt{Z}}_1\,|\,{\tt{Z}}_2)$
and where ${\tt{H}}^T = ({\tt{H}}_1^T \,| \, {\tt{H}}_2^T)$.

\begin{proof}[Proof of Theorem~\ref{OneToOneThree}]
Denote by ${\cal{B}}$ the set of bases of ${\cal{G}}$, and by ${\cal{T}}$ the set of extremal classical processing relations of form Eq.~(\ref{TO7a}), (\ref{TO7b}), specified by the triple $(I_{\text{gauge}}, O_{\text{comp}}, \{ {\tt{H}},{\tt{R}},{\tt{T}},{\tt{Z}}\})$. Then, the mapping $h: {\cal{B}} \longrightarrow {\cal{T}}$ exists and is a bijection. (1) Existence of $h$: By the normal form Eq.~(\ref{GNF}) of ${\cal{G}}= (\Phi||S)$, for a given basis $B({\cal{G}})$ the sets $I_{\text{gauge}}$ and $O_{\text{comp}}$ are extracted as follows. A qubit $i$ is in $I_{\text{gauge}}$ if and only if the corresponding column of $\Phi$ appears in $B({\cal{G}})$. A qubit $a$ is in $O_{\text{comp}}$ if and only if the corresponding column of $S$ does not appear in $B({\cal{G}})$. Knowing $I_{\text{gauge}}$ and $O_{\text{comp}}$, $\{ {\tt{H}},{\tt{R}},{\tt{T}},{\tt{Z}}\}$ is uniquely determined via Theorem~\ref{OneToOneOne}. (2) Surjectivity of $h$: By definition of ``extremal''. (3) Injectivity of $h$: given $I_{\text{gauge}}$ and $O_{\text{comp}}$, $B(G)=\left( \Phi|_{I_{\text{gauge}}}\,| \, S_{O_{\text{comp}}}\right)$ is unique. 
\end{proof}

\begin{proof}[Proof of Theorem~\ref{OneToOneTwo}]
We divide the proof into the following steps: i) First, with $\text{rk}\,H =\text{rk}\, Z$, the processing relations stem from an extremal pair $I_{\text{gauge}}$,$O_{\text{comp}}$. We show that given an extremal pair $I_{\text{gauge}}$, $O_{\text{comp}}$, 
the matrices $\tt{H}$, $\tt{R}$, ${\tt{T}}$ and $\tt{Z}$ 
are uniquely determined by the classical  processing relations~Eq.~(\ref{TO7a},\ref{TO7b}). ii) From
these matrices we can derive the corresponding normal form of the resource state uniquely.  Let  us 
denote this by $\mathcal{N}$. In iii) and iv) we show that the following diagram commutes, which establishes the  equivalence of normal forms of all extremal pairs. 
\begin{eqnarray}
\begin{CD}
(I_{\text{gauge}}, O_{\text{comp}})  @>{\Lambda}>> (I_{\text{gauge}}', O_{\text{comp}}')\\
@VV \tt{T},\tt{H},\tt{Z},\tt{R}V  @VV{\tt{T}',\tt{H}',\tt{Z}',\tt{R}'}V \\
\mathcal{N}@> {{M(\Lambda)}}>> \mathcal{N}'
\end{CD}\label{eq:cd}
\end{eqnarray}
In Eq.~\eqref{eq:cd}, $M(\Lambda)$ is a transformation on the normal form $\mathcal{N}$, dependent on $\Lambda$.
Because of their independence, we consider the transformations $I_{\text{gauge}} \rightarrow I_{\text{gauge}}'$ and $ O_{\text{comp}}\rightarrow O_{\text{comp}}'$ separately in iii) and iv) respectively.

\begin{compactenum}[i)]
\item
Extracting $\tt{H}$, $\tt{R}$, ${\tt{T}}$ and $\tt{Z}$: Given any set  of inputs $I$ and outputs $O$ by Theorem~\ref{Extr}, we know that there always 
exists an   extremal pair $(I_{\text{gauge} }, O_{\text{comp}})$,  where $I_{\text{gauge} } \subseteq I, O_{\text{comp}}\subseteq O$. 
By Eq.~(\ref{Tshape}),  $\tt{T}$ is uniquely specified by $T$.  By the assumption of the classical output being optimal, $Z$ can be brought into the unique normal form of Eq.~(\ref{ZN}) by left multiplication with an invertible matrix  and (permuting the columns if necessary), ${\tt{Z}}$ can be extracted. 
%
%
The matrix $H$ may not appear in Eq.~(\ref{TO7a}) in its normal form Eq.~(\ref{HRTZshape}), 
nonetheless
for an invertible matrix $\Lambda$, the vector $\textbf{q}$ specifying the measurement bases is invariant under $H \longrightarrow H \Lambda$, $\textbf{g} \longrightarrow \Lambda^{-1}\textbf{g}$.  
By definition of $I_{\text{gauge}}$, every qubit in $I_{\text{gauge}}$ can be individually gauged 
with respect to $I_{\text{gauge}}$, $O_{\text{comp}}$. 
Therefore, (up to row permutations), we can choose $\Lambda$ such that 
\begin{equation}
  \label{HNF}
H \Lambda =  \left(\begin{array}{c} I \\ \hline {\tt{H}} \end{array}\right),
\end{equation}
where the row split is between $I_{\text{gauge}}$ (upper) and $(I_{\text{gauge}})^c$ (lower). $\Lambda$ and ${\tt{H}}$ are unique.
The classical output $\textbf{o}$ is invariant under the transformation $R \longrightarrow R \Lambda',\, \textbf{g} \longrightarrow (\Lambda')^{-1}\textbf{g}$. However, since Eqs.~(\ref{TO7a}) and (\ref{TO7b}) refer to $\textbf{g}$ in the same basis, $\Lambda'$ is now fixed: $\Lambda'=\Lambda$. Then, ${\tt{R}} = R\Lambda$, with the $\Lambda$ of Eq.~(\ref{HNF}).

\item Assembling the normal form: Using the unique matrices $\tt{H}$, $\tt{R}$, $\tt{T}$, $\tt{Z}$, by Lemma~\ref{NF}, we can now assemble the normal form Eq.~(\ref{GNF}) of the stabilizer generator matrix for the resource state $|\Psi\rangle$. By assumption of Theorem~\ref{OneToOneTwo}, the matrices $H$, $R$, $T$, $Z$ describe a valid computation, and the normal form Eq.~(\ref{GNF}) derived from them must thus yield a valid description of a quantum state. In particular, all Pauli operators specified by the rows of ${\cal{G}}(\{\tt{H}, \tt{R}, \tt{T}, \tt{Z}\})$ in Eq.~(\ref{GNF}) must commute. (The rows are independent by design of the normal form.) We have thus constructed a description of $|\Psi\rangle$. Since $\tt{H}$, $\tt{R}$, $\tt{T}$, $\tt{Z}$ are unique, so is $|\Psi\rangle$. \medskip

We now proceed to construct the stabilizer ${\cal{S}}(|\Psi\rangle)$ from the processing relations when $I_{\text{gauge}}$ and $O_{\text{comp}}$ are not specified. From the classical processing relation Eq.~(\ref{TO7a}) we can still extract the input set $I$ and the output set $O$,  by testing which rows and columns of $T$ identically vanish. Then, the possible choices for $I_{\text{gauge}}$ and $O_{\text{comp}}$ are limited to $I_{\text{gauge}} \subseteq I$ and $O_{\text{comp}} \subseteq O$ by definition.

\noindent
\item Equivalence under $\Lambda:I_{\text{gauge}}\rightarrow I_{\text{gauge}}'$. In order to prove this
we rely on the slightly variant version of the normal form of ${\cal{G}}(|\Psi\rangle)$ as shown in Eq.~(\ref{GN5})  which also includes additional detail about the correction operators.
As noted above, the different normal forms for $H$ can be interconverted by  right-multiplication with an invertible matrix $\Lambda$ (change of basis for $\textbf{g}$), i.e. $H, {\tt{R}} \longrightarrow H\Lambda, R\Lambda$, where $H^T=(I|{{\tt{H}}^T})$. Under such a transformation, the upper part of the normal form Eq.~(\ref{GNF}) for ${\cal{G}}(|\Psi\rangle)$ remains unchanged, and the lower part is transformed $(H^T||0|R^T) \longrightarrow  \Lambda^T(H^T||0|{{\tt{R}}^T})$. Invertible row transformations on ${\cal{G}}(|\Psi\rangle)$ leave $|\Psi\rangle$ unchanged, and $|\Psi\rangle$ is thus independent on the precise choice of $I_{\text{gauge}} \subseteq I$.  

\item  Equivalence under $\Lambda:O_{\text{comp}}\rightarrow O_{\text{comp}}'$.
Proving that a different choice of $O_{\text{comp}}$ does not change the stabilizer state is a little
more complicated. We proceed in the following manner. From Theorem~\ref{OneToOneThree},
we know that $I_{\text{gauge}}\cup O_{\text{comp}}^c $ are the bases of a matroid.
Therefore for two distinct
computational output sets $O_{\text{comp}}$ and $O_{\text{comp}}''$, there exists another computational
output set $O_{\text{comp}}'= O_{\text{comp}}\setminus \{i \} \cup \{j \}$, where 
$i\in O_{\text{comp}}\setminus O_{\text{comp}}''$ and $j\in O_{\text{comp}}''\setminus O_{\text{comp}}$.
Therefore, it  suffices if we show that the stabilizer state does not change if we change the computational output set from $O_{\text{comp}}$ to $O_{\text{comp}}'$.
Assume that the classical relation for the computational output set $O_{\text{comp}}$ is given as
\begin{eqnarray}
\mathbf{o} = Z \mathbf{s} + R  \mathbf{g}  = ({\tt{Z}}_1\, | {\tt{Z}}_2 | I )\mathbf{s} + {\tt{R}} \mathbf{g}.
\end{eqnarray}
where the column split of $Z$ is between $O_{\text{comp}}^c$, $\Delta O$, and $O_{\text{comp}}$. The 
corresponding normal form (with the column split in $\sigma_{\phi}$-part $I_{\text{gauge}}|I_{\text{gauge}}^c$ ) is
\begin{eqnarray}
  \left(\begin{array}{c|c||c|c|c}
    0&  \tilde{\tt{T}}^T & I & 0 & {\tt{Z}}_1^T \\ \hline
     0 & 0 & 0 &  I & {\tt{Z}}_2^T \\ \hline
     I & \mbox{\tt{H}}^T & 0 &  0 & {\tt{R}}^T
  \end{array}\right),\label{eq:nform2}
\end{eqnarray}
where ${\tt{T}}^T= \left(\begin{array}{c}  \tilde{\tt{T}}^T\\\hline 0 \end{array} \right)$.
Suppose that we transform $O_{\text{comp}}$ to $O_{\text{comp}} \setminus \{i \}\cup \{j\} $,
where $i\in O_{\text{comp}}$  and $j\in \Delta O $. Without loss of generality assume that $i$ is the 
last column of ${\tt{Z}}_2$. (It cannot be an all zero column because, then it would not be possible for it to be in 
$O_{\text{comp}}$.) Let ${\tt{Z}}_1= \left(\begin{array}{c}x^T\\\hline Z_A \end{array}\right)$ and 
${\tt{Z}}_2 = \left(\begin{array}{c|c} a^T&1\\\hline Z_B&b  \end{array}\right)$. 
Then  $\Lambda = \left(\begin{array}{c|c} 1&0 \\\hline b &I\end{array}\right)$ acting on 
$Z$ achieves the transformation $O_{\text{comp}}$ to $O_{\text{comp}}'$. 
\begin{eqnarray}
\Lambda ({\tt{Z}}_1\,|\,{\tt{Z}}_2| I )  &=& \left(\begin{array}{c|c|c|c|c}x^T &a^T&1&1&0\\\hline Z_A+bx^T& Z_B+ba^T&0&b& I  \end{array}\right)\\ 
&\sim& \left(\begin{array}{c||c|c||c|c}x^T &a^T&1&1&0\\\hline Z_A+bx^T& Z_B+ba^T&b&0& I  \end{array}\right)  = ({\tt{Z}}_1'\,||\,{\tt{Z}}_2'|| I)
\end{eqnarray}
We claim that this same transformation can be effected by row transformations of ${\cal{G}}(|\Psi\rangle)$. First let us focus on the middle set of rows in ${\cal{G}}(|\Psi\rangle)$, namely the 
correction operators for $\Delta O$. Then
acting by $M(\Lambda)= \left(\begin{array}{c|c} I & a \\\hline 0 &1\end{array}\right)$ gives us
\begin{eqnarray}
M(\Lambda)(0 || 0 |I | {\tt{Z}}_2^T) & = & \left( \begin{array}{c||c|c|c|c|c}0&0&I & a & 0 & Z_B^T+ab^T\\\hline0&0& 0 & 1& 1& b^T\end{array}\right)\\
&\sim& \left( \begin{array}{c||c|c|c|c|c}0&0&I & 0 & a & Z_B^T+ab^T\\\hline 0&0&0 & 1& 1& b^T\end{array}\right)  = (0||0|I | {\tt{Z}}_2'^T ).\label{eq:lastRow}
\end{eqnarray}
Now if take the last row in Eq.~\eqref{eq:lastRow}, namely $(0|| 0 | 0 |1|b^T) = c$ 
 and add  $x c$ to the top set of rows in Eq.~\eqref{eq:nform2} we  obtain 
\begin{eqnarray}
 \left( \begin{array}{c|c||c|c|c|c|c}0&\tilde{\tt{T}}^T&I&0 & x & 0 & Z_A^T+xb^T\end{array}\right)
&\sim& \left( \begin{array}{c|c||c|c|c}0&\tilde{\tt{T}}^T&I&0  & {\tt{Z}_1}'^T\end{array}\right)
\end{eqnarray}
showing the equivalence of $\tt{Z}$ and $\tt{Z}'$. The equivalence of $\tt{R}$ and $\tt{R}'$
under $\Lambda$ can be shown in exactly the same fashion as for ${\tt{Z}}_1$ and ${\tt{Z}}_1'$.
\end{compactenum}
This concludes the proof that the extremal classical relations completely determine the stabilizer
state. 
\end{proof}

\section{Invariance under the gauge transformations}
\label{GTO}

\subsection{Gauge transformations and temporal order}

In this section we provide a different angle at Theorem~\ref{OneToOneOne}, namely we show that the gauge transformations Eq.~(\ref{G1}) impose severe constraints on the possible temporal orders for a given stabilizer state and set of measurement planes.

The gauge transformations Eq.~(\ref{G1}) act on $\textbf{q}$ and $\textbf{s}$, and can therefore have a non-trivial effect on the classical processing relation Eq.~(\ref{TO7a}), $\textbf{q} = T\textbf{s} + H \textbf{g}$. We now study this effect. Since the transformations Eq.~(\ref{G1}) are caused by the insertion of stabilizer operators into the state overlap in Eq.~(\ref{corr1}), they do not change the physical situation. Therefore, the temporal relations before and after any such transformation must be equally valid, although not necessarily identical. By insertion of the stabilizer operator into the overlap $\langle \Phi_{\text{loc}}|\Psi\rangle$, the stabilizer of $|\Psi\rangle$ and the the sets $I_{\text{gauge}}$, $O_{\text{comp}}$ do not change. Therefore, by Theorem~\ref{OneToOneOne}, the matrices $T$ and $H$ do not change. Thus, besides $\textbf{q}$ and $\textbf{s}$, all that can change in the relation Eq.~(\ref{TO7a}) under a transformation Eq.~(\ref{G1}) is $\textbf{g}$. The following two viewpoints are always equivalent: (A) The relation $\textbf{q} = f_{\textbf{g}}(\textbf{s})$, under the action Eq.~(\ref{G1}) of a $G_K$ on $(\textbf{s},\textbf{q})$ is changed into an equivalent such relation $\textbf{q} = f_{\textbf{g}'}(\textbf{s})$, with $G_K: \textbf{g} \longrightarrow \textbf{g}'$. (B) The relation $\textbf{q} = \tilde{f}(\textbf{s},\textbf{g})$ {\em{remains invariant}} under all transformations $G_K$, acting on the triple $(\textbf{s},\textbf{q},\textbf{g})$. We choose the latter viewpoint.

We now infer the action of the transformations $G_K$ on $\textbf{g}$. Without loss of generality we assume that the relations Eq.~(\ref{TO7a}) are given with $H$ in its normal form Eq.~(\ref{HRTZshape}),
$$
\textbf{q} = T \textbf{s} +  \left( \begin{array}{c} I \\ \hline {\tt{H}}  \end{array} \right) \textbf{g} \mod 2.
$$ 
Furthermore, we assume that the pair $I_{\text{gauge}}$, $O_{\text{comp}}$ is extremal. Since $I_{\text{gauge}} \subseteq I$ by definition, for each $i \in I_{\text{gauge}}$, $q_i$ only depends on $\textbf{g}$ but not on the measurement outcomes $\textbf{s}$, $q_i = g_i$. Now, the correction operators $K(a)$, $a \in (O_{\text{comp}})^c$ derived from the normal form Eq.~(\ref{GNF}) of ${\cal{G}}(|\Psi\rangle)$, $K(a)|_{I_{\text{gauge}}}$  has no $\sigma_\phi$-part. Therefore, the corresponding transformations $G_K$ do not flip $q_i$, for all $i \in I_{\text{gauge}}$. In order to preserve the relation $q_i=g_i$, they thus leave $\textbf{g}$ unchanged. Now consider the other stabilizer generators, $\overline{K}(i)$, $i \in I_{\text{gauge}}$, obeying the conditions Eq.~(\ref{CorrGauge}). By construction, $G_{\overline{K}(i)}$ flips $q_i$ but no other $q_j$, for $i\neq j \in I_{\text{gauge}}$. Hence, to preserve the relations $q_i=g_i$, it must also flip $g_i$, but no other $g_j$, $i \neq j \in I_{\text{gauge}}$. Thus, for a stabilizer element $K=\bigotimes_{a \in \Omega}(\sigma_s^{(a)})^{v_a}(\sigma_\phi^{(a)})^{w_a}$,
\begin{equation}
  \label{G2}
  G_K: \; \mathbf{g} \longrightarrow \mathbf{g} \oplus \mathbf{w}|_{I_{\text{gauge}}}.
\end{equation}  
Therein, we have assumed that the basis choice for $\textbf{g}$ is such that the matrix $H$ appearing in Eq.~(\ref{TO7a}) is of normal form Eq.~(\ref{HRTZshape}).

We have now fully specified the action  of $G_K$ on the triple $(\textbf{q}, \textbf{s}, \textbf{g})$, c.f. Eq.~(\ref{G1}), (\ref{G2}). MBQCs satisfy the invariance condition
\begin{equation}
  \label{GK1}
    \textbf{q} = T\textbf{s} + H\textbf{g}  \mod 2 \Longleftrightarrow   G_K(\textbf{q}) = T\, G_K(\textbf{s}) + H\, G_K(\textbf{g}) \mod 2, \, \forall K \in {\cal{S}}(|\Psi\rangle).
\end{equation}
It is evident that the requirement (\ref{GK1}) of invariance of the processing relations (\ref{TO7a}) under the gauge transformations poses constraints on the possible matrices $T$  and ${\tt{H}}$. In fact, as we show below, given $O_{\text{comp}}$ the matrices $T$ and ${\tt{H}}$ are uniquely specified uniquely by the above invariance condition.\medskip

To check the invariance condition Eq.~(\ref{GK1}) in a specific case, we return to our 3-qubit cluster state example of Section~\ref{OTO}. We consider the effect of the transformations induced by generators $K_1=\sigma_\phi^{(1)}\sigma_s^{(2)}$, $K_2=\sigma_s^{(1)}\sigma_{\phi}^{(2)}\sigma_s^{(3)}$ and $K_3=\sigma_s^{(2)}\sigma_\phi^{(3)}$ on the processing relations Eq.~(\ref{TrelEx}). As noted earlier, $I_{\text{gauge}}=\{1\}$. Then, with Eqs.~(\ref{G1}) and (\ref{G2}),
\begin{equation}
  \label{Gex3}
  \begin{array}{rlll}
    G_{K_1}:& \textbf{q} \longrightarrow \textbf{q} \oplus (1,0,0)^T,& \textbf{s} \longrightarrow \textbf{s} \oplus (0,1,0)^T, & g_1 \longrightarrow g_1 \oplus 1, \\
    G_{K_2}:& \textbf{q} \longrightarrow \textbf{q} \oplus (0,1,0)^T,& \textbf{s} \longrightarrow \textbf{s} \oplus (1,0,1)^T, & g_1 \longrightarrow g_1,\\
    G_{K_3}:& \textbf{q} \longrightarrow \textbf{q} \oplus (0,0,1)^T,& \textbf{s} \longrightarrow \textbf{s} \oplus (0,1,0)^T, & g_1 \longrightarrow g_1.
  \end{array}
\end{equation}
It is easily checked that the relation Eq.~(\ref{TrelEx}) is invariant under the transformations  $G_{K_1}$, $G_{K_2}$ and $G_{K_3}$ of Eq.~(\ref{Gex3}). However, if the transformations are restricted to $\textbf{q}$, $\textbf{s}$, the relation Eq.~(\ref{TrelEx}) is no longer invariant under the transformation induced by $K_1$.\medskip

We now return to the general case and show that, given the set $O_{\text{comp}}$ and the action Eq.~(\ref{G1}), (\ref{G2}) of the gauge transformations on the triple $(\textbf{q},\textbf{s},\textbf{g})$, the invariance condition Eq.~(\ref{GK1}) uniquely specifies the classical processing relations Eq.~(\ref{TO7a}) for the adaption of measurement bases.

Recall that we write the stabilizer generator matrix for $|\Psi\rangle$ in the $\sigma_\phi/\sigma_s$-basis as ${\cal{G}}(|\Psi\rangle) = (\Phi||S)$. Then, for the stabilizer generator $K_a \in {\cal{S}}(|\Psi\rangle)$ corresponding to the $a$-th row of ${\cal{G}}(|\Psi\rangle)$, with Eq.~(\ref{G1}) the action of the gauge transformation $G_{K_a}$ on $\textbf{s}$, $\textbf{q}$ is
\begin{equation}
  G_{K_a}:\; \textbf{s} \longrightarrow \textbf{s} \oplus \text{row}_a(S),\; \textbf{q} \longrightarrow \textbf{q} \oplus \text{row}_a(\Phi).
\end{equation}
With Eq.~(\ref{G2}), the action of $G_{K_a}$ on $\textbf{g}$ is
\begin{equation}
   G_{K_a}: \longrightarrow \textbf{g} \oplus \text{row}_a(\Phi)|_{I_{\text{gauge}}}.
\end{equation}
Here, $\text{row}_a(\Phi)|_{I_{\text{gauge}}}$ denotes $\text{row}_a(\Phi)$ restricted to $I_{\text{gauge}}$. Then, the condition Eq.~(\ref{GK1}) for invariance of $\textbf{q} =T \textbf{s} + H \textbf{g}$ under $G_{K_a}$ becomes 
$$
\text{row}_a(\Phi) = T  \text{row}_a(S) + H \text{row}_a(\Phi)|_{I_{\text{gauge}}} \mod 2.
$$
This condition must hold for all stabilizer generators $K_a$ simultaneously, hence
\begin{equation}
  \label{SymCon}
  \Phi^T = T S^T + H \Phi^T|_{I_{\text{gauge} \times \Omega}} \mod 2.
\end{equation} 
By definition, the qubits in $I_{\text{gauge}}$ have empty backward cones, and the qubits in $O_{\text{comp}}$ have empty forward cones, hence $T$ is of the form
$$
T =  \left(\begin{array}{c|c} 0 & 0 \\ \hline  {\tt{T}} & 0 \end{array}\right),
$$
where the column split is $(O_{\text{comp}})^c|O_{\text{comp}}$ and the row split is $I_{\text{gauge}}|(I_{\text{gauge}})^c$, c.f. Eq.~(\ref{Tshape}). By right-multiplication of relation Eq.~(\ref{SymCon}) with a suitable matrix, we transform $\Phi^T|_{I_{\text{gauge} \times \Omega}}$ into a matrix of form $(I|0)$ where the column split is between $I_{\text{gauge}}$ and $I_{\text{gauge}}^c$. By definition of $I_{\text{gauge}}$, such a transformation is always possible. Under the same transformation,
\begin{equation}
  \Phi^T \longrightarrow \left( \begin{array}{c|c} I & 0 \\ \hline \Phi_1 & \Phi_2 \end{array}\right),\;\; S^T|_{(O_{\text{comp}})^c \times \Omega} \longrightarrow (S_1|S_2).
\end{equation}
Inserting the above into Eq.~(\ref{SymCon}), we find that $H$ must be of normal form Eq.~(\ref{HRTZshape}), $H = \left( \begin{array}{c} I \\ \hline {\tt{H}} \end{array}\right)$, and
\begin{equation}
\begin{array}{rcl}
  \Phi_1 &=& {\tt{T}} S_1 + {\tt{H}} \mod 2,\\
  \Phi_2 &=& {\tt{T}} S_2 \mod 2.
\end{array}
\end{equation}
Now, $S_2$ must be an invertible matrix. This is the condition that, by definition of $O_{\text{comp}}$, every measurement outcome in $(O_{\text{comp}})^c$ is correctable. Then,
\begin{equation}
  {\tt{T}} = \Phi_2 {S_2}^{-1} \mod 2,\; {\tt{H}} = \Phi_1 + \Phi_2 {S_2}^{-1} S_1 \mod 2.
\end{equation} 
Hence the relation $\textbf{q} = T\textbf{s} + H \textbf{g}$ is uniquely specified.

\subsection{Gauge transformations and computational output}
 
In addition to Eq.~\ref{GK1}, we also require invariance of the classical output under the transformations Eq.~(\ref{G1}), (\ref{G2}),
\begin{equation}
  \label{GK2}
 \textbf{o} = Z\textbf{s} + R\textbf{g} \mod 2 = Z\, G_K(\textbf{s}) + R\,G_K(\textbf{g}) \mod 2, \, \forall K \in {\cal{S}}(|\Psi\rangle).
\end{equation}
Like Eq.~(\ref{GK1}), Eq.~(\ref{GK2}) is a determinism constraint. If for a single output bit $o$ the relation $o = \textbf{z}^T\textbf{s} + \textbf{r}^T\textbf{g}$ is not invariant under all gauge transformations Eq.~(\ref{G1}), (\ref{G2}), then the value of $o$ is guaranteed to be random, and thus useless as readout bit of a computation. Specifically,

\begin{Lemma} 
\label{Det}
Assume an MBQC where the relation $\textbf{q} =T\textbf{s}+H\textbf{g}$ is invariant under the gauge transformations Eq.~(\ref{G1}), (\ref{G2}), but an output bit $o$ exists whose defining relation $o = \textbf{z}^T\textbf{s} + \textbf{r}^T\textbf{g}$ is not invariant under the action of $G_K$ for some $K \in {\cal{S}}(|\Psi\rangle)$. Then, the value of $o$ is completely random, independent of the choice of measurement angles. 
\end{Lemma}

{\em{Proof of Lemma~\ref{Det}.}} For simplicity, consider first the special case where $K \in {\cal{S}}(|\Psi\rangle)$ acts trivially on $\textbf{g}$, $G_K (\textbf{g}) = \textbf{g},\, \forall g$. We may then write $o = \sum_{i \in J}s_i + c$ for an offset $c = \textbf{r}^T\textbf{g}$. We call the string $\textbf{s}|_J$ of measurement outcomes on $J$ even (odd) if it has even (odd) weight. We denote the local post-measurement state on qubit $a$ by $|\varphi_a, s_a,q_a(\textbf{s}, \textbf{g})\rangle$, where $\varphi_a$ is the measurement angle, $s_a$ the measurement outcome and $q_a$ specifies the chosen measurement basis. 

Under the transformation $G_K$, $\textbf{s} \longrightarrow \textbf{s} \oplus \Delta \textbf{s}_K$, where, by assumption, $\Delta \textbf{s}_K|_J$ is odd. Now, the probability of outputting $o=c$ is
\begin{equation}
\label{rand}
  \begin{array}{rcl}
  p(o=c) &= & \displaystyle{\sum_{\textbf{s}|_J = \text{even}} \left|\left( \bigotimes_{a \in \Omega} \mbox{}_a\langle \varphi_a, s_a, q_a(\textbf{s},\textbf{g})|\right) |\Psi\rangle \right|^2}\\
  & =&  \displaystyle{\sum_{\textbf{s}|_J = \text{even}} \left|\left( \bigotimes_{a \in \Omega} \mbox{}_a\langle \varphi_a, s_a, q_a(\textbf{s},\textbf{g})|\right) K |\Psi\rangle \right|^2}\\
& =& \displaystyle{\sum_{\textbf{s}|_J = \text{even}} \left|\left( \bigotimes_{a \in \Omega} \mbox{}_a\langle \varphi_a, s_a \oplus \Delta s_{K,a}, q_a\oplus \Delta q_{K,a}|\right) |\Psi\rangle \right|^2}\\
 & =& \displaystyle{\sum_{\textbf{s}|_J = \text{even}} \left|\left( \bigotimes_{a \in \Omega} \mbox{}_a\langle \varphi_a, s_a \oplus \Delta s_{K,a}, q_a(\textbf{s} \oplus \Delta \textbf{s}_K,\textbf{g})|\right) |\Psi\rangle \right|^2}\\
 & =&  \displaystyle{\sum_{\textbf{s}|_J = \text{odd}} \left|\left( \bigotimes_{a \in \Omega} \mbox{}_a\langle \varphi_a, s_a, q_a(\textbf{s},\textbf{g})|\right) |\Psi\rangle \right|^2}\\
&=& p(o=\overline{c}).
  \end{array}
\end{equation} 
Thus, $p(o=c) = p(o=\overline{c})=1/2$. Note that in transitioning from the third to the fourth line of Eq.~(\ref{rand}) we have used the invariance property Eq.~(\ref{GK1}), i.e., the assumption that the adaption of measurement bases is deterministic.

In the general case, $G_K: \textbf{s} \longrightarrow \textbf{s} \oplus \Delta \textbf{s}_K,\; \textbf{g} \longrightarrow \textbf{g} \oplus \Delta \textbf{g}_K$. We note that we can choose any gauge fixing $\textbf{g}$, and thus $p(o=c)= \displaystyle{\frac{1}{2^{|I_{\text{gauge}}|}} \sum_{\textbf{g}} \sum_{\textbf{s}|_J = \text{even}} \left|\left( \bigotimes_{a \in \Omega} \mbox{}_a\langle \varphi_a, s_a, q_a(\textbf{s},\textbf{g})|\right) |\Psi\rangle \right|^2}$. By an argument analogous to the above we then find $p(o=0) = p(o=1) = 1/2$. $\Box$

\section{What is time in MBQC?}
\label{Flip}

It has been established earlier that `time' in measurement based quantum computation is a transitive binary relation. If a given MBQC is to be deterministically runnable, then we require a strict partial ordering among the measurement events. If, on the other side, in pursuing an analogy with general relativity we allow for closed time-like curves, then antisymmetry and irreflexivity  are not required. Given this setting, in the previous sections of this paper we have analyzed the constraints on MBQC temporal relations imposed by a group of gauge transformations. 

In this section, we provide another angle at the question of what time in MBQC is. Namely, we introduce a second group of symmetry transformations on MBQCs which is related to local complementation \cite{LC1} - \cite{LC3}. These transformations in general change the MBQC on which they act, but leave temporal relations invariant. The possible temporal relations in MBQC thus label representations of this group.  

\subsection{Flipping measurement planes}

According to Eq.~(\ref{Obs}), for any qubit $a$ in a given resource state, the local observable measured to drive the computation is
$$
O_a[q_a]= \cos \varphi_a\, \sigma_\phi^{(a)}  +  (-1)^{q_a}\sin \varphi_a\, \sigma_{s\phi}^{(a)}.
$$
Therein, $q_a$ is a linear function of the measurement outcomes $\{s_b,\, b \in \Omega\}$, c.f. Eq.~(\ref{TO7a}). 

Let's see what happens if we use a different rule for the adjustment of measurement bases, namely
\begin{equation}
\label{lObsPr}
O_a'[q_a]= (-1)^{q_a} \cos \varphi_a\, \sigma_\phi^{(a)}  + \sin \varphi_a\, \sigma_{s\phi}^{(a)},
\end{equation}
That is, if $q_a=1$, to obtain $O'_a[1]$ we are flipping the observable $O_a[0]$ about the $\sigma_{s\phi}$-axis rather than the $\sigma_\phi$-axis. Comparing Eqs. (\ref{Obs}) and (\ref{lObsPr}), we find that 
\begin{equation}
\label{ObsId}
O'_a[q_a] \equiv (-1)^{q_a}O_a[q_a],
\end{equation}
independent of the measurement angle $\varphi_a$. The measurements of $O'_a[q_a]$ and  $O_a[q_a]$ are always in the same basis for the same $q_a$, and the measured eigenvalues differ by a factor of $(-1)^{q_a}$. 

We call the transformation $\tau_s[i]: O_i[q_i] \longrightarrow O'_i[q_i]$ {\em{flipping of the measurement plane at qubit $i$}}. On the elementary degrees of freedom, namely the resource state $|\Psi\rangle$, the Pauli observables $\sigma_s^{(i)}$, $\sigma_{s\phi}^{(i)}$ (action on $\sigma_s^{(i)}$ is implied), and the measurement angle $\varphi_i$, the flipping $\tau_s[i]$ acts as
\begin{equation}
\label{tau_s}
\tau_s[i]: \begin{array}{lcl} \sigma_\phi^{(i)} &\longleftrightarrow& \sigma_{s\phi}^{(i)},\\ \varphi_i &\longrightarrow& (-1)^{q_i}\frac{\pi}{2} - \varphi_i,\\
|\Psi\rangle & \longrightarrow & |\Psi\rangle, \end{array}
\end{equation}
The action of $\tau_s[i]$ on the Pauli operators $\sigma^{(j)}$ and the measurement angles $\varphi_j$, for $j\neq i$ is trivial. All our considerations are independent of the values of the measurement angles. In particular, the second part of the transformation Eq.~(\ref{tau_s}) does not affect temporal order.

We now discuss the effect of the flipping $\tau_s[a]$, $a \in \Omega$ on a given MBQC. Denote the measurement outcome of a measurement of $O'_a(q_a)$ by $s'_a$. By Eq.~(\ref{ObsId}), the following two measurement procedures are always equivalent. (I) Measuring $O_a[q_a]$ and outputting $s_a$, and (II) Measuring $O'_a[q_a]$ and outputting $s'_a + q_a \mod 2$. We may call the device that performs Procedure I a $\phi$-box, and the device that performs Procedure II a $s\phi'$-box. Then, 
\begin{center}
\includegraphics[width=6cm]{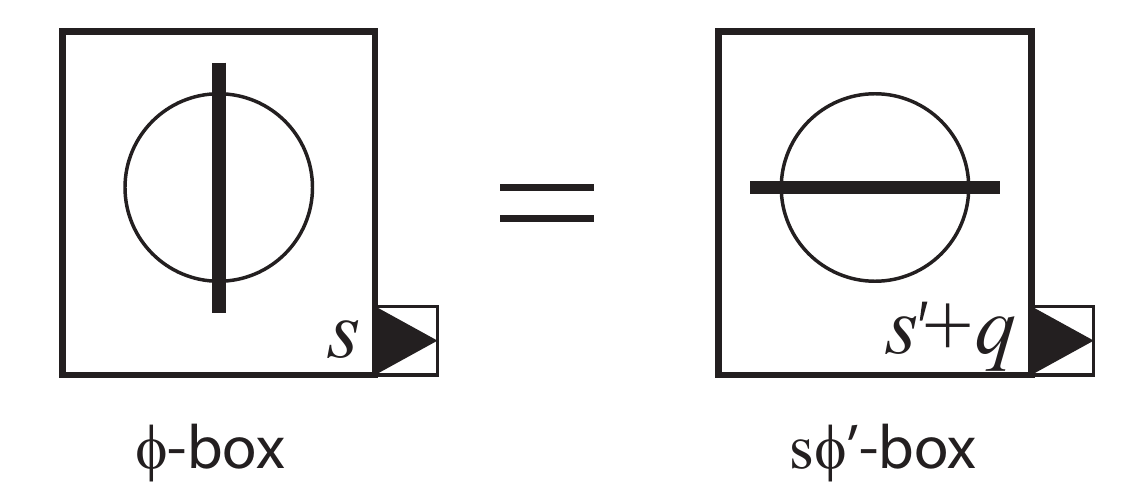}.
\end{center}
The prime in the $s\phi'$-box accounts for the fact that not the measurement outcome $s'_a$ itself is outputted, but rather the locally post-processed value $s'_a+q_a \mod 2$. Now, instead of outputting $s'_a+q_a$, the device at $a$ may only output $s'_a$ (that is an $s\phi$-box), and the classical post-processing relations for the adaption of measurement bases are modified accordingly, i.e. $s_a \longrightarrow s_a'+q_a \mod 2$. Can the resulting relations again be written in a form $\textbf{q} = T'\textbf{s}' + H'\textbf{g}$?  
 
We now attempt transforming a $\phi$-box into an $s\phi$-box at qubit $a$. The vectors of measurement outcomes $\textbf{s}$ and $\textbf{s}'$ are related via $\textbf{s} = \textbf{s}' + e_ae_a^T\textbf{q}$. Inserting this relation into Eq.~(\ref{TO7a}), we obtain 
\begin{equation}
\label{star}
(I+T \textbf{e}_a \textbf{e}_a^T)\textbf{q} = T\textbf{s}'  + H\textbf{g} \mod 2.
\end{equation} 

Case I: $T_{aa}=0$. Physically, this means that the measurement basis at the flipped qubit $a$ does not depend on the measurement outcome at $a$, before the transformation. Multiplying Eq.~(\ref{star}) with $\textbf{e}_a^T$ from the left yields $q_a = \textbf{e}_a^TT\textbf{s}'+\textbf{e}_a^TH\textbf{g} \mod 2$. Inserting back into Eq.~(\ref{star}), we obtain
\begin{equation}
  \label{Tlc}
  T' = T + T \textbf{e}_a \textbf{e}_a^T T \mod 2.
\end{equation}
Likewise,
\begin{equation}
  \label{TlcCons}
  H' = H \oplus T \textbf{e}_a \textbf{e}_a^TH, \; Z' = Z \oplus  Z \textbf{e}_a \textbf{e}_a^T T,\; R' = R \oplus  Z \textbf{e}_a \textbf{e}_a^T H.  
\end{equation}
Eqs.~(\ref{Tlc}) and (\ref{TlcCons}) completely describe the effect of the flipping $\tau_s[a]$ of the measurement plane at qubit $a$ on the classical processing relations Eq.~(\ref{TO7a}), (\ref{TO7b}). 

{\em{Remark:}} If the matrices $H,R,T,Z$ are given their normal form Eq.~(\ref{HRTZshape}) wrt the pair $I_{\text{gauge}},O_{\text{comp}}$ then the flipping of the measurement plane at any vertex $a$ with $T_{aa}=0$ leaves this normal form intact. We can therefore state a transformation rule equivalent to Eqs.~(\ref{Tlc}), (\ref{TlcCons}) for the matrices $\tt{H},\tt{R},\tt{T},\tt{Z}$. This rule is, in fact, simpler. We again consider the composite matrix $T_{\text{ext}}$, c.f. Eq.~(\ref{TppL}), 
$$
T_{\text{ext}} =  \left(\begin{array}{c|c|c|c} 0 & 0 & 0 & 0\\ I & 0 & 0 &0 \\ \hline {\tt{H}} & {\tt{T}}  & 0 & 0\\ \hline {\tt{R}} & {\tt{Z}} & I & 0 \end{array}\right),
$$
The effect of flipping of the measurement plane at $a$ then is
\begin{equation}
\label{TlcExt}
T_{\text{ext}} \longrightarrow T'_{\text{ext}} = T_{\text{ext}} + T_{\text{ext}}\, \textbf{e}_a \textbf{e}_a^T\, T_{\text{ext}} \mod 2.
\end{equation}
Thus, the rule is just the same as Eq.~(\ref{Tlc}) for the original influence matrix $T$.\medskip

Case 2: $T_{aa}=1$. In this case, the measurement basis at $a$ does depend on the measurement outcome at $a$. This is an example for a closed time-like curve (only involving the measurement device at $a$), and an obstacle to deterministic runnability. Now, the matrix $I\oplus T \textbf{e}_a \textbf{e}_a^T$ on the left side in Eq.~(\ref{star}) is not invertible, $\textbf{e}_A^T(I \oplus T \textbf{e}_a \textbf{e}_a^T) =0$. Hence, the relation (\ref{star}) can not be solved for $\textbf{q}$ in this case. There is no relation Eq.~(\ref{TO7a}) with the same sets $I_{\text{gauge}}$, $O_{\text{comp}}$ before and after flipping.\medskip

We now discuss the consequences of flipping measurement planes for the above two cases.

\subsection{Flipping measurement planes and local complementation}

We now return to the above Case 1, namely when flipping of a measurement plane yields a computation with a new relation Eq.~(\ref{TO7a}). Note that the computation before and after the flip generate the same output distribution. Flipping a $\phi$-box into an $s\phi'$-box is an equivalence transformation, only based on the operator identity Eq.~(\ref{ObsId}). Changing an $s\phi'$-box into an $s\phi$-box is again an equivalence transformation, provided it can be carried out. 

The influence matrices $T$ and $T'$ before and after the flipping, respectively, are in general not the same, c.f. Eq.~(\ref{Tlc}). However, $T$ and $T'$ still generate the same temporal order, as we now show.
\begin{Lemma}
\label{TOpres}
Be $T$ an influence matrix with $T_{ii}=0$. Then, $T$ and $T' = T \oplus T \textbf{e}_i \textbf{e}_i^T T$ generate the same temporal relation under transitivity.  
\end{Lemma}

{\em{Proof of Lemma~\ref{TOpres}.}} Let's introduce a shorthand $a \rightarrow c$ for $c \in fc(a)$ (meaning that the measurement outcome at $a$ influences the measurement basis at $c$). Now, we have to show that $e \prec_T f \Longleftrightarrow e \in \prec_{T'} f$, for any $T'$ generated from $T$ by the transformation Eq.~(\ref{Tlc}).

(I) ``$\Longrightarrow$'': Assume that $e \prec_T f$. Then there exists a sequence of measurement events $e \rightarrow m_1 \rightarrow m_2 \rightarrow ..\rightarrow a \rightarrow c \rightarrow .. \rightarrow f$. Can we break the arrow $a \rightarrow c$, say? To investigate this, let us rewrite the transformation rule Eq.~(\ref{Tlc}) for the flip $\tau_s[i]$ as
\begin{equation}
\label{Tpr3}
\tau_s[i]: \begin{array}{lclr}
  \textbf{fc}(a) &\longrightarrow& \textbf{fc}(a) \oplus \textbf{fc}(i),& \mbox{if } i \in fc(a),\\ 
  \textbf{fc}(a) &\longrightarrow& \textbf{fc}(a),& \mbox{if } i \not\in fc(a),
\end{array}
\end{equation}
Case 1: $a \rightarrow i$ before the transformation $\tau_s[i]$. Then, $\textbf{fc}'(a) = \textbf{fc}(a) \oplus \textbf{fc}(i)$. Since $T_{ii}=0$ by assumption, $a \rightarrow i$ after the transformation $\tau_s[i]$. Sub-case 1a: $i \rightarrow c$ before the transformation $\tau_s[i]$. Since $T_{ii}=0$ ($i \not\in fc(i)$), $i \rightarrow c$ after the transformation $\tau_s[i]$. Thus $a \rightarrow i \rightarrow c$ after the transformation, and hence $a \prec_{T'} c$.
Sub-case 1b: $i \not\rightarrow c$ before $\tau_s[i]$. Then, $a \rightarrow c$ remains after the transformation. Case 2:  $a \not \rightarrow i$ before the transformation $\tau_s[i]$. Then $a \rightarrow c$ after $\tau_s[i]$.
Thus, in all cases $a \prec_{T'} c$, and therefore $e \prec_{T'} f$.

 (II) ``$\Longleftarrow$'': From Eq.~(\ref{Tpr3}), $\tau_s[i]^2 = I$. $\Box$\medskip

Apply a series of transformations Eq.~(\ref{Tlc}) on an initial influence matrix $T$ with vanishing diagonal part may produce an influence matrix with a non-vanishing diagonal part. Thus, the application of the transformation Eq.~(\ref{Tlc}) is restricted. To circumvent this problem, we introduce a modified transformation
\begin{equation}
  \label{Tlc2}
  \tilde{\tau}[i]: T \longrightarrow T' =  T + T\textbf{e}_i \textbf{e}_i^TT +{\cal{D}}(T\textbf{e}_i \textbf{e}_i^TT) \mod 2.
\end{equation} 
Clearly, this transformation takes influence matrices with vanishing diagonal part to influence matrices with vanishing diagonal part, and thereby avoids the problem of restricted applicability of transformation Eq.~(\ref{Tlc}). Note that the transformation Eq.~(\ref{Tlc2}) has the form of local complementation, albeit the influence matrix $T$ that it acts on will in general not be symmetric.

But what is the physical significance of transformation Eq.~(\ref{Tlc2})? The only additional effect of the transformation $\tilde{\tau}[i]$ over $\tau_s[i]$ is the cancelling of the diagonal part of the influence matrix after the transformation, c.f. the last term in Eq.~(\ref{Tlc2}). This can be achieved by a local unitary that exchanges $\sigma_{s\phi}\leftrightarrow \sigma_s$ on a respective qubit. The action of $\tilde{\tau}[i]$ on the elementary degrees of freedom therefore is
\begin{equation}
\label{Tlc3}
\tilde{\tau}[i]: \begin{array}{lcll} \sigma_\phi^{(i)} &\longleftrightarrow& \sigma_{s\phi}^{(i)},\\
 \varphi_i &\longrightarrow& (-1)^{q_i}\frac{\pi}{2} - \varphi_i, \\ 
\sigma_s^{(j)} & \longleftrightarrow & \sigma_{s\phi}^{(j)}, & \forall j \in fc(i) \cap bc(i),\\
|\Psi\rangle &\longrightarrow & |\Psi\rangle. 
\end{array}
\end{equation}
We find that the local measured operators for all qubits $j \in bc(i) \cap fc(i)$ change in a way that cannot be accommodated by a change of the respective measurement angle. For those qubits, the new measured observables  lie in a different equatorial plane of the Bloch sphere. Therefore,  the transformation Eq.~(\ref{Tlc3}), unlike the transformation Eq.~(\ref{Tlc}), does {\em{not}} necessarily map a given computation onto itself. What it does, however, is mapping a given computation to a computation with the same temporal relation. 

\begin{Lemma}
\label{TI}
Be $T$ an influence matrix with $T_{ii}=0$. Then, $T$ and $T' = T \oplus  T \textbf{e}_i \textbf{e}_i^T T \oplus {\cal{D}}(T \textbf{e}_i \textbf{e}_i^T T)$ generate the same temporal relation under transitivity.  
\end{Lemma}

\noindent{Proof of Lemma~\ref{TI}.} Assuming the initial influence matrix has vanishing diagonal part, we split the transformation $T \longrightarrow T \oplus  T \textbf{e}_i \textbf{e}_i^T T \oplus {\cal{D}}(T \textbf{e}_i \textbf{e}_i^T T)$ into two steps, namely $T \longrightarrow T' = T \oplus  T \textbf{e}_i \textbf{e}_i^T T$ and $T' \longrightarrow T'' = T' \oplus {\cal{D}}(T')$. By Lemma~\ref{TOpres}, $T$ and $T'$ generate the same temporal order. Now assume that $T'_{kk}=1$ for some $k \in \Omega$, $k \neq i$. This requires that $T_{ik}=T_{ki}=1$. Then, we also have $T_{ik}'=T_{ki}'=1$ and $T_{ik}''=T_{ki}''=1$. Thus, $k \prec_{T'} k$ and $k \prec_{T''} k$. The closed time-like curve involving $k$ is not changed by setting $T''_{kk}=0$. All other relations trivially remain unaffected by the transformation $T' \longrightarrow T''$. $\Box$  \medskip

We may now consider the orbit $O_T$ of a particular influence matrix $T$ under local complementation Eq.~(\ref{Tlc3}). Every local complementation $\tilde{\tau}[i]$ permutes the elements of $O_T$. Therefore, we have a homomorphism $\tilde{\tau}[i] \longrightarrow P[i]$ where $P[i]$ is a permutation matrix of size $|O_T| \times |O_T|$. The matrices $\{P[i]|\, i \in \Omega\}$ generate a representation of the local complementation group. The representation acts on a given orbit $O_T$ of influence matrices, and all $T \in O_T$ give rise to the same temporal relation under transitivity. Therefore, this temporal relation is a label of the given representation. All possible temporal relations in MBQC arise as labels of representations of the local complementation group.

\section{Breaking up closed time-like curves}
\label{BRKctc}

\subsection{Closed time-like curves of length 1}

In this section we return to the case of $T_{ii}=1$ before flipping the measurement plane at qubit $i$. Such a qubit $i$ cannot be in $I_{\text{gauge}}$, since $I_{\text{gauge}} \subseteq I$ by definition. If $T_{ii}=1$ then $i \in bc(i)$. The backward cones of all qubits in $i$ are empty by definition of $I$, however. Likewise, $i \not \in O_{\text{comp}}$. If $T_{ii}=1$ then $i \in fc(i)$. However, $fc(a)=\emptyset$ for all $a \in O_{\text{comp}}$. Thus, there is only one case to consider, namely $i \in (I_{\text{gauge}})^c\cap (O_{\text{comp}})^c$.

In this case, there exists a correction operator $K(i)$ for qubit $i$ before the flipping, $K(i) = \sigma_{s\phi} \otimes K(i)|_{\Omega \backslash i}$. After flipping at $i$, this operator turns into
\begin{equation}
  \tau_s[i](K(i)) = \sigma^{(i)}_\phi \otimes K(i)|_{\Omega \backslash i}=:\overline{K}'(i).
\end{equation} 
That is, the operator $\tau_s[i](K(i))$ resulting from flipping at $i$ is a gauge type operator, c.f. Eq.~(\ref{CorrGauge}). Thus, the flipping transformation $\tau_s[i]$ (when $T_{ii}=1$) enlarges $I_{\text{gauge}}$ by one qubit,
$$
\tau_s[i]:\; I_{\text{gauge}}  \longrightarrow I_{\text{gauge}} \cup \{i\},\;\; \mbox{if } T_{ii}=1.
$$
Furthermore, after the flipping at $i$ there no longer is a correction operation for qubit $i$, hence
$$
\tau_s[i]:\; O_{\text{comp}}  \longrightarrow O_{\text{comp}} \cup \{i\},\;\; \mbox{if } T_{ii}=1.
$$
This has two consequences. First, the forward cone of $i$ becomes empty. In particular $T_{ii}=0$ after the flipping. Thus, the closed time-like curve consisting of qubit $i$ has been removed. Second, an additional bit of optimal classical output is being created by the flipping at $i$. What does that output bit signify?

Recall that before the flipping at $i$, the rule for adjusting the measurement basis at $i$ is
$$
q_i \begin{array}{c} !\vspace{-2mm} \\= \vspace{-2mm} \\ \mbox{ }\end{array} s_i + \sum_{j \in J\backslash i}s_j \mod 2, 
$$
for some set $J \subseteq \Omega$. Here, we have dropped a constant offset $\textbf{h}^T\textbf{g}$ on the r.h.s. The symbol ``!'' above the equality means that equality is a requirement for the correctness of the computation, but it cannot be deterministically implemented. As follows from Eq.~(\ref{ObsId}), the measurement outcomes before and after the flip, $s_i$ and $s_i'$ are related via $s_i = s_i' \oplus q_i$. For all the other qubits, $s_j' = s_j$. Substituting this into the above relation, we obtain
\begin{equation}
\label{s_tau}
s_i' + \sum_{j \in J\backslash i}s_j' \mod 2 \begin{array}{c} !\vspace{-2mm} \\= \vspace{-2mm} \\ \mbox{ }\end{array} 0, \;\; \forall q_a' \in \mathbb{Z}_2.
\end{equation}
Thus, the additional output bit $o_i= s_i' + \sum_{j \in J\backslash i}s_j' \mod 2$ is a flag. If $o_i=0$ then the computation succeeded, and if $o_i=1$ then it did not.

Now suppose that the problem solved by the given MBQC is in NP. Then, this flag bit is not necessary. The remaining output may be efficiently checked for correctness anyway. Thus, one may safely discard the extra bit $o_i$ of output. Not post-selecting on $o_i=0$ can, if anything, only increase the success probability of the computation. We thus arrive at

\begin{Lemma}
  \label{CTCrem}
  Be ${\cal{M}}_1$ an MBQC with a classical output $\textbf{o}$ and influence matrix $T$ such that $T_{ii}=1$, i.e., ${\cal{M}}_1$ has a closed time-like curve involving a single qubit $i \in \Omega$. Be ${\cal{M}}_2$ the MBQC with the same classical output $\textbf{o}$, obtained from ${\cal{M}}_1$ by flipping the measurement plane at $i$. Then, the closed time-like curve of $i$ in ${\cal{M}}_1$ is removed in ${\cal{M}}_2$. Furthermore, if ${\cal{M}}_1$ solves a problem in the complexity class NP with probability $p$, then ${\cal{M}}_2$ solves the same problem with probability $\geq p$.
\end{Lemma}

{\em{Remark:}} Lemma~\ref{CTCrem} does not guard against the inefficiencies of post-selection, in particular if multiple CTCs of length 1 are being removed. While the success probability after removing the CTCs is guaranteed not to be smaller than for the original computation with the CTCs (which can only be executed using post-selection), neither it is provably significantly larger.\medskip 
 
{\em{Event horizons.}} Let us consider the flow of information between qubit $i$ whose measurement plane has been flipped and the other qubits. Before the flip (MBQC ${\cal{M}}_1$ of Lemma~\ref{CTCrem}), $i \in (O_{\text{comp}})^c \cap (I_{\text{gauge}})^c$. After the the flip (MBQC ${\cal{M}}_2$ of Lemma~\ref{CTCrem}), $i \in O_{\text{comp}} \cap I_{\text{gauge}}$. In ${\cal{M}}_2$, since $i \in I_{\text{gauge}}$, no information for the adaption of measurement basis is flowing into site $i$ from the other sites. Likewise, since $i \in O_{\text{comp}}$, no information for the adaption of measurement bases is flowing out of site $i$. Finally, because of the normal form Eq.~(\ref{HRTZshape}), the measurement outcome $s_i$ appears in only one readout bit. This readout bit is $o_i$ as given in l.h.s. of Eq.~(\ref{s_tau}), which is precisely the bit of classical output that can be discarded if the problem solved by the quantum computation is in NP. If $o_i$ is discarded, then no information is flowing out of the site $i$ at all. Thus, in summary, from the viewpoint of classical processing, qubit $i$ in ${\cal{M}}_2$ becomes entirely disconnected from the computation. It vanishes behind the MBQC counterpart of an event horizon. 

\begin{figure}
\begin{center}
  \includegraphics[width=15cm]{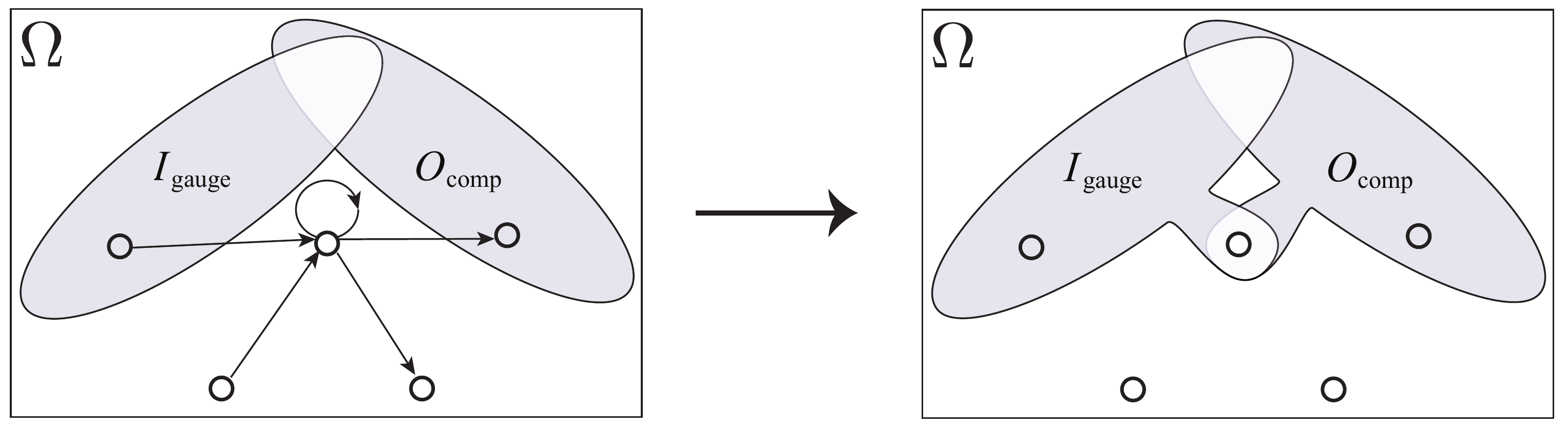}
  \caption{\label{CTCbreak} Breaking a closed time-like curve of length 1. The looped qubit becomes an element of $I_{\text{gauge}} \cap O_{\text{comp}}$ after flipping the measurement plane. As such both its forward cone ($O_{\text{comp}} \subseteq O$), and backward cone  ($I_{\text{gauge}} \subseteq I$) must be empty.}
\end{center}
\end{figure} 

\subsection{Closed time-like curves of length $\geq 2$}

A closed time-like curve ${\cal{L}} = \{1,2,..,l\}$ of length $l \geq 2$ is a set of qubits such that the relations $1 \longrightarrow 2$, $2 \longrightarrow 3$, .. ,$i \longrightarrow i+1$, .., $l-1 \longrightarrow l$, $l \longrightarrow 1$ hold. We call a closed time-like curve ${\cal{L}}$ minimal if no proper subset of ${\cal{L}}$ is a closed time-like curve. Minimal closed time-like curves are of the form
\begin{equation}
  fc(i)\cap {\cal{L}} = \{i+1\},\; 1 \leq i \leq l,\;\; (\mbox{where } l+1\equiv 1). 
\end{equation}
That is, for a minimal closed time-like curve, no arrows other than the defining ones are present.
Proof: Assume $\exists k>1$ such that $i+k \in fc(i)$. Then, ${\cal{L}}'=\{1,2,..,i-1,i,i+k,i+k+1,..,l\}$ is also a closed time-like curve, and a proper subset of ${\cal{L}}$. Hence, ${\cal{L}}$ is not minimal. Contradiction. $\Box$

Regarding the embedding of ${\cal{L}}$ in the overall temporal structure, we have
\begin{equation}
{\cal{L}} \subset (I_{\text{gauge}})^c \cap (O_{\text{comp}})^c.
\end{equation}
Proof: $\forall i \in {\cal{L}}$, $i-1 \in bc(i)$, hence $bc(i) \neq \emptyset$. Hence, $i \not \in I_{\text{gauge}}\subset I$. Furthermore, $\forall i \in {\cal{L}}$, $i+1 \in fc(i)$. Hence, $fc(i) \neq \emptyset$. Hence, $i \not\in O_{\text{comp}} \subset O$. $\Box$

Now, the CTC ${\cal{L}}$ can be broken up between qubits 1 and l. This time, no transformation $\tau_s$ is necessary. The starting point ``1" for the labelling is of course arbitrary. To break up the CTC, we enlarge the sets of gauge input and computational output,
\begin{equation}
\label{IOtrans}
 I'_{\text{gauge}} = I_{\text{gauge}} \cup \{1\},\; O_{\text{comp}}'= O_{\text{comp}} \cup \{l\},
\end{equation}
and modify the correction and gauge operators
\begin{equation}
\label{Ktrans}
\begin{array}{rcll}
  K(l) & \longrightarrow & \overline{K}'(1):=K(l),\vspace{1mm}\\
  K(a) & \longrightarrow & K'(a) := \left\{
  \begin{array}{ll} 
    K(a),& \mbox{if}\; 1 \not \in fc(a)\\
    K(a)K(l),& \mbox{if}\; 1 \in fc(a)
  \end{array}\right.,
  \forall a \in (O_{\text{comp}})^c\backslash \{l\},\vspace{1mm}\\
  \overline{K}(i) & \longrightarrow & \overline{K}'(i) := \left\{
  \begin{array}{ll} 
    \overline{K}(i),& \mbox{if}\; \overline{K}(i)|_1 =I^{(1)}\\
    \overline{K}(i)K(l),& \mbox{if}\; \overline{K}(i)|_1 =\sigma_s^{(1)}
  \end{array}\right.,
  \forall i \in I_{\text{gauge}}.
\end{array}
\end{equation}
The first line in Eq.~(\ref{Ktrans}) says that $K(l)$ is re-interpreted as $\overline{K}(1)$. Hence, qubit $l$ no longer has a correction operation and thus becomes a member of $O_{\text{comp}}$, as required in Eq.~(\ref{IOtrans}). The second
line in Eq.~(\ref{Ktrans}) ensures that $K'(a)|_1 = I^{(1)}$ for all $a \in (O_{\text{comp}})^c$ such that $bc'(i) = \emptyset$. Thus, $1 \in  I_{\text{gauge}}' \subset I'$, as required in Eq.~(\ref{IOtrans}). The third line in Eq.~(\ref{Ktrans}) makes the gauge operators compatible with the new normal form for ${\cal{G}}(|\Psi)$ (based on $I'_{\text{gauge}}$ and $O'_{\text{comp}}$), and has no effect on temporal order. Since, after the reshuffling of the gauge and correction operators according to Eq.~(\ref{Ktrans}), $1 \in I_{\text{gauge}}'$ and $l \in O_{\text{comp}}'$ and we have successfully broken up the CTC ${\cal{L}}$ between 1 and $l$. Since the computational output set has been enlarged by one qubit, $l$, we have one bit $o_l$ of additional classical output. What does it signify? - When retracing the transformations Eq.~(\ref{Ktrans}) in the normal form Eq.~(\ref{GNF}) of the stabilizer generator matrix ${\cal{G}}(|\Psi\rangle)$, we find that the matrix ${\tt{Z}}$ is updated according to
\begin{equation}
  {\tt{Z}} \longrightarrow \left(\begin{array}{c} \textbf{bc}(1)^T\backslash \{l\}\\ \hline {\tt{Z}} \backslash \{l\} \end{array} \right),
\end{equation}
where $\textbf{bc}(1)\backslash \{l\}$ is the characteristic vector of $bc(1)$, restricted to $O_{\text{comp}}\backslash \{l\}$, and ${\tt{Z}}\backslash \{l\}$ is the matrix obtained from ${\tt{Z}}$ by deleting column $l$. Thus we find that the new output bit $o'_l$ is
\begin{equation}
\label{addlout}
o'_l:= s_l + \sum_{a \in bc(1)\backslash \{l\}}s_a \mod 2.
\end{equation}
We compare this to the relation $q_1 = s_l+ \sum_{a \in bc(1)\backslash \{l\}}s_a \mod 2$ before the transformation Eq.~(\ref{Ktrans}), and find the following interpretation of $o'_l$: The temporal relations before the transformation Eq.~(\ref{Ktrans}) contain a closed time-like curve and can therefore only be implemented probabilistically. So one may as well assume $q_1 = 0$, perform the measurement and later check whether the relation between $q_1$ and $\textbf{s}$ was obeyed. With Eq.~(\ref{addlout}), this check amounts to $o'_l = 0$, and $o'_l$ is thus a flag for the correctness of the computation. If the computational problem solved is in NP then we can afford to discard this extra bit $o'_l$ of output. The solution can be efficiently checked anyway. In no longer post-selecting on $o'_l = 0$, we retain all the `good' cases while admitting additional ones. The success
probability of the algorithm does not decrease. We thus arrive at

\begin{Lemma} \label{CTCl2}
Be ${\cal{M}}_1$ an MBQC with a classical output $\textbf{o}$ and influence matrix $T$ giving rise to a minimal closed time-like curve ${\cal{L}} = \{1,2,..,l\}$ of length $\geq 2$. Be ${\cal{M}}_2$ the MBQC with the same classical output $\textbf{o}$, obtained from ${\cal{M}}_1$ by enlarging $I_{\text{gauge}} \longrightarrow  I_{\text{gauge}} \cup \{1\}$, $O_{\text{comp}} \longrightarrow O_{\text{comp}} \cup \{l\}$. Then, the closed time-like curve ${\cal{L}}$ is not present in ${\cal{M}}_2$. Furthermore, if ${\cal{M}}_1$ solves a problem in the complexity class NP with probability $p$, then ${\cal{M}}_2$ solves the same problem with probability $\geq p$.
\end{Lemma}

We still need to check that by the transformation Eq.~(\ref{Ktrans}) of the correction and gauge operators we do not create additional closed time-like curves in exchange for removing the minimal CTC ${\cal{L}}$. In this regard, note that the middle line of Eq.~(\ref{Ktrans}) amounts to
$$
fc(a) \longrightarrow fc'(a) = fc(a) \oplus fc(l), \forall a \in (O'_{\text{comp}})^c \mbox{ with } 1 \in fc(a).
$$
Therefore, for any $b \in (I'_{\text{gauge}})^c$, $b \in fc'(a)$ only if $b \in fc(a)$ or $b \in fc(l)$. Thus, $b \in fc'(a)$ only if
$a \prec b$ before the enlargement of $I_{\text{gauge}}$, $O_{\text{comp}}$. Therefore, no new closed time-like curves are being created.

\subsection{Removal of closed time-like curves}

Combining the above two cases of CTCs of length 1 and $\geq 2$, respectively, we can remove CTCs repeatedly until none remains. The process must stop, since at some point every qubit is in the gauge input or/and the computational output set. We thus arrive at the following

\begin{Theorem}\label{CTCrem5} Be ${\cal{M}}_1$ an MBQC solving a problem in the complexity class NP with probability $p$. Then, all closed time-like curves can be removed from the influence matrix $T$ of ${\cal{M}}_1$, resulting in a
deterministically runnable MBQC ${\cal{M}}_2$ which solves the same problem with probability $\geq p$.
\end{Theorem}

{\em{Remark:}} Theorem~\ref{CTCrem5} does not guard against the inefficiencies of post-selection. While the success probability after removing the CTCs is guaranteed not to be smaller than for the original computation
with the CTCs (which can only be executed using post-selection), we cannot prove that it is significantly larger.

\section{Conclusions and outlook}
\label{Concl}

In this paper we have suggested MBQC as a toy model for spacetime emerging from quantum mechanics. We have analyzed the constraints on MBQC temporal order that arise from the quantum mechanical randomness of measurement, in combination with the computational requirement of preventing this randomness from affecting the quantum information processing. We have identified two groups of symmetry transformations, a gauge transformation which leaves every MBQC---including the temporal relations---invariant, and a second transformation, related to local complementation, which only preserves the temporal relations. We have shown that for any given resource stabilizer state $|\Psi\rangle$ and set $\Sigma$ of measurement planes, the compatible temporal relations all arise from the bases of a matroid encompassing $|\Psi\rangle$,$\Sigma$. Finally, we have identified event horizons as a further piece of the phenomenology of general relativity which has a counterpart in MBQC.    

At this point, we are led to ask the following questions:
\begin{enumerate} 
\item{We introduced a group of gauge transformations Eq.~(\ref{G1},\ref{G2}), and a group symmetry transformations Eq.~(\ref{Tlc2}), generated by flipping measurement planes. Both transformations preserve MBQC temporal orders. Can the two groups be unified?}
\item{Some of the temporal relations admitted by the matroid ${\cal{G}}(|\Psi\rangle)$ contain closed time-like curves. Given a stabilizer state $|\Psi\rangle$ and set of measurement planes $\Sigma$, can we find a similar algebraic (or other) structure which comprises only the partial orders? Can the partial order of measurements with the smallest set $O_{\text{comp}}$ be efficiently computed?} 
\item{We have established that all temporal relations free of self-loops appearing in MBQC arise as representations of the local complementation group. However, these representations are in general reducible. If we decompose them into irreps, where does temporal order go?} 
\item{\label{Qspace}In proposing MBQC as a toy model for spacetime, here we have focussed on the temporal part. What is a suitable notion of {\em{space}} that can be associated with a resource state and a set of measurement planes in MBQC?}
\item{\label{QDet}In MBQC, the link between the randomness in quantum mechanical measurement and temporal order is the principle that the randomness of measurement outcomes should not affect the logical processing. In a context more general than quantum computation, what could this principle be replaced by?}
\end{enumerate}
{\em{Remark regarding Question~\ref{Qspace}:}} With Theorem~\ref{OneToOneTwo} we arrive at a situation closely resembling Malament's theorem \cite{conTC}. The latter states that for a continuous spacetime manifold, the temporal order of spacetime events determines the spacetime metric up to a conformal factor, i.e., determines the topology of spacetime. In MBQC, the classical processing relations---which are essentially temporal information encoded in the influence matrix $T$, plus the extra matrices $H$, $R$, $Z$ related to the boundary sets $I_{\text{gauge}}$ and $O_{\text{comp}}$---determine the entire computation up to the measurement angles. To further pursue an analogy with Malament's theorem it is therefore desirable to have a characterization of `space' in MBQC that admits a discussion in terms of topology. Cellular complexes thus seem a suitable choice. Indeed, certain stabilizer states considered as resource states admit the characterization in terms of a cellular complex, such as the planar code state \cite{Kita} (two-dimensional complex) as well as computationally universal two-dimensional cluster states (three-dimensional \cite{BR06} and four-dimensional complexes).

{\em{Remark regarding Question~\ref{QDet}:}} A possible criticism of MBQC as a toy model for a quantum spacetime is that the laws which govern it, namely the classical processing relations Eq.~(\ref{TO7a}), (\ref{TO7b}) are not laws of nature, but only rules imposed by the requirement of shielding the processed quantum information from the randomness of measurement. As such, the processing relations may be obeyed or violated at will by an operator running the MBQC.

To this we respond that if the operating unit for a measurement-based quantum computer was a dedicated device, it would not have the freedom to violate the classical processing rules. For it, they {\em{would}} be the laws of nature. They would follow straight from Newton's axioms if the device was mechanical, and from Maxwell's equations if it was electrical. The processing rules could only be violated by a conscious being, such as a human or trained animal.

Such beings, bound to the pull of gravity and the consequences of the no-cloning theorem by their belonging to this `real' universe, but entitled by their free will to disobey man-made regulations, can violate the processing rules of MBQC precisely because they have an existence outside it. Therefore, a theory of MBQC should not be required to describe them. Neither quantum mechanics nor the theory of gravity make statements about hypothetical objects that jump in and out of spacetime.\medskip 

To conclude, we would like to recall the main idea underlying this work.  In attempts to unify the theory of general relativity with quantum mechanics, one may take the viewpoint that spacetime is not an independent construct, but rather a consequence of the laws of quantum mechanics. Once this assertion is spelled out, the natural next step is to identify the key quantum property which yields a mechanism for generating temporal order, and to illustrate this mechanism in a toy model. Measurement based quantum computation provides such a toy model. Therein, the key property which leads to temporal order is the inherent randomness of quantum measurement.

\end{document}